\documentclass[english, 12pt, letter]{article}
\usepackage{amsmath,amssymb,amsthm}
\usepackage{natbib}
\usepackage[nodisplayskipstretch]{setspace}
\usepackage{bm}
\usepackage{threeparttable}
\usepackage{graphicx}
\usepackage{enumerate}
\usepackage{commath}			
\usepackage{bibentry}
\usepackage{booktabs}
\usepackage[titletoc,title]{appendix}
\usepackage{multirow}
\usepackage{arydshln}
\usepackage{tabularx}
\usepackage{dsfont}
\usepackage{bibunits}
\usepackage[font={footnotesize}]{subcaption}
\usepackage[left=1in, right=1in, top=1in, bottom=1in]{geometry}

\usepackage[colorinlistoftodos,textsize=small]{todonotes}
\allowdisplaybreaks

\defaultbibliography{thebib.bib}
\defaultbibliographystyle{jaestyle2}

\usepackage[
	breaklinks,
	colorlinks=true,
	pagebackref=false,
	pdfauthor={},%
	pdftitle={},%
	citecolor=blue,
	linkcolor=blue,
	urlcolor=blue
	]{hyperref}      

\newtheorem{prop}{Proposition}

\newtheorem{thm}{Theorem}

\newtheoremstyle{remark2}{1ex}{1ex}%
      {}
      {}
      {\bf}
      {.}
      {5pt}
      {\thmname{#1}\thmnumber{ #2}\thmnote{ \slshape{(#3)}}} 
\theoremstyle{remark2}
\newtheorem{rem}{Remark}
\newtheorem{defn}{Definition}

\usepackage{babel}
\input{ee.sty}

\@ifundefined{bibfont}{  } 
    {  }

\makeatletter
\renewenvironment{proof}[1][\bfseries\proofname]{\par
   \pushQED{\qed}%
   \normalfont \topsep6\p@\@plus6\p@\relax
   \trivlist
   \item[\hskip\labelsep
     #1\@addpunct{:}]\ignorespaces
}{%
   \popQED\endtrivlist\@endpefalse
}
\makeatother

\newcommand{\Comments}{1}
\newcommand{\mynote}[2]{\ifnum\Comments=1\textcolor{#1}{#2}\fi}
\newcommand{\mytodo}[2]{\ifnum\Comments=1%
  \todo[linecolor=#1!80!black,backgroundcolor=#1,bordercolor=#1!80!black]{#2}\fi}

\ifnum\Comments=1               
\fi

\ifnum\Comments=1               
\fi

\newcommand{\p}{\operatorname{P}}

\newcommand{\D}{\,\mathrm{d}}
\newcommand{\loc}{\operatorname{loc}}

\newcommand{\dGDP}{\Delta\text{GDP}}

\begin{document}

\baselineskip18pt
\renewcommand\floatpagefraction{.9}
\renewcommand\topfraction{.9}
\renewcommand\bottomfraction{.9}
\renewcommand\textfraction{.1}
\setcounter{totalnumber}{50}
\setcounter{topnumber}{50}
\setcounter{bottomnumber}{50}
\abovedisplayskip1.5ex plus1ex minus1ex
\belowdisplayskip1.5ex plus1ex minus1ex
\abovedisplayshortskip1.5ex plus1ex minus1ex
\belowdisplayshortskip1.5ex plus1ex minus1ex

\title{On Testing Equal Conditional Predictive Ability Under Measurement Error\thanks{The authors would like to thank Tobias Fissler, Christoph Hanck, Onno Kleen and Andrew Patton for their very insightful comments on an earlier version of this manuscript.  The authors are also grateful to seminar participants at the University of Bonn and the Heidelberg Institute for Theoretical Studies. The first author thankfully acknowledges support of the German Research Foundation (DFG) through project HO 6305/1-1.}}

\author{
	Yannick Hoga\thanks{\textbf{Corresponding author:} Faculty of Economics and Business Administration, University of Duisburg-Essen, Universit\"atsstra\ss e 12, D--45117 Essen, Germany, \href{mailto:yannick.hoga@vwl.uni-due.de}{yannick.hoga@vwl.uni-due.de}.}
\and
	Timo Dimitriadis\thanks{Alfred Weber Institute of Economics, Heidelberg University, Bergheimer Str. 58, D--69115 Heidelberg, Germany, and Heidelberg Institute for Theoretical Studies, \href{mailto:timo.dimitriadis@awi.uni-heidelberg.de}{timo.dimitriadis@awi.uni-heidelberg.de}.} 
}

\date{\today}
\maketitle

\begin{abstract}
\noindent Loss functions are widely used to compare several competing forecasts. However, forecast comparisons are often based on mismeasured proxy variables for the true target. We introduce the concept of \textit{exact robustness to measurement error} for loss functions and fully characterize this class of loss functions as the Bregman class.
For such exactly robust loss functions, forecast loss differences are on average unaffected by the use of proxy variables and, thus, inference on conditional predictive ability can be carried out as usual.  Moreover, we show that more precise proxies give predictive ability tests higher power in discriminating between competing forecasts. Simulations illustrate the different behavior of exactly robust and non-robust loss functions. An empirical application to US GDP growth rates demonstrates that it is easier to discriminate between forecasts issued at different horizons if a better proxy for GDP growth is used. \\

\noindent \textbf{Keywords:} Equal Predictive Ability, Forecasting, Hypothesis Testing, Measurement Error  \\
\noindent \textbf{JEL classification:} C12 (Hypothesis Testing), C52 (Model Evaluation, Validation, and Selection), C53	(Forecasting and Prediction Methods)

\end{abstract}


\begin{bibunit}

\section{Motivation}

\doublespacing   

Due to the central role of forecasts in economic policy, business, climate research and beyond, forecast comparisons have a long tradition. Such comparisons rely on a (statistically or economically motivated) loss function that measures the loss as a function of the issued forecast and the realization of the target variable.
Since the seminal contribution of \citet{DM95}, tests of equal predictive ability (EPA) have played a central role in comparing competing forecasts. \citet{GW06} extend EPA tests by introducing tests of equal \textit{conditional} predictive ability (ECPA), where the conditioning is, e.g., on current economic conditions. The null hypothesis of ECPA tests is that the conditional mean of the forecast losses are identical.


E(C)PA tests are studied extensively under estimation error in the forecasts (see, e.g., \citealp{Wes96, ClarkMcCracken2001, Pat20}). However, the effect of measurement error in the observed target variable has not received as much attention. Some exceptions---to be discussed below---are the works of \citet{Pat11}, \citet{Laurent2013} and \citet{LP18}.
Nonetheless, measurement error is present in many economic and financial time series. Examples in economics include the gross domestic product (GDP) \citep{Aea16}, inflation rates \citep{CM09,FS16}, and job earnings \citep{AS13}. 
In finance, the conditional variance can only be approximated by the squared return or high-frequency measures such as realized volatility \citep{Aea13}. 
Examples beyond economics and finance include, among others, meteorological applications such as the measurement of precipitation or wind speeds \citep{Fer17}.

As a consequence, many forecast comparisons are carried out with approximated, mismeasured target variables, also called \textit{proxies}.
In such a case, the forecast losses---as measured by some loss function---may be systematically different from those obtained using the actual, but latent, target variable. Hence, differences in predictive ability may be clouded by the use of such proxies. In this paper, we derive conditions under which ECPA tests can be validly carried out if only some (conditionally unbiased) proxy for the target variable is available. If several alternative proxies are available, we further derive conditions which proxy entails the most powerful tests.

To do so, we define a loss function to be \textit{exactly robust to measurement error} if the (conditional) expectation of the forecast loss differences is unchanged when using the proxy instead of the true target variable. Since most of the literature on forecast evaluation is concerned with univariate quantities \citep{Gne11,Pat11}, it is worth stressing that the target variable and the forecasts may be multivariate here. Some work on characterizing strictly consistent loss functions for the specific multivariate mean functional can be found in \citet{BGW05}, \citet{Laurent2013} and \citet{FK15}.

Our first main contribution is to characterize the loss functions that are exactly robust to measurement error. We show that the class of exactly robust loss functions coincides with the Bregman loss functions \citep[Def.~1]{BGW05}, and the class to which \cite{Pat11} refers as ``robust'' loss functions in the univariate case. 
This implies that only conditional mean forecasts can be compared robustly in ECPA tests.
While this is of course rather restrictive, for many economic variables that are measured with error, the conditional mean is precisely the object of interest; e.g., the conditional mean of GDP growth or inflation, or the conditional mean of squared asset returns (which commonly coincides with the conditional return variance).
The importance of the conditional mean is further underscored by the prevalence in economics of (V)ARMA-type models, which are designed to dynamically model the conditional mean.

The studies most closely related to our first contribution are those of \citet{Pat11} (in the univariate case) and \citet{Laurent2013} (in the multivariate case), both of which are devoted to the specific task of comparing conditional variance forecasts for financial returns using high-frequency (HF) proxies.
\citet{Pat11} (later on extended by \citet{Laurent2013}) characterizes loss functions that give consistent \textit{relative} rankings when only a conditionally unbiased proxy for the conditional variance is available. Our exact robustness is a stronger requirement than \citeauthor{Pat11}'s \citeyearpar{Pat11} ordering robustness. While under exact robustness two forecasts have the same conditional predictive ability for the true target and the proxy, ordering robustness merely implies that the ordering of the two forecasts is preserved when a proxy is used instead of the true target; however, the magnitude in predictive ability may be changed. For instance, the average forecast loss differences may be large for the true target, yet very small when the proxy is used. Thus, while the ranking (in population) is preserved, it may be much harder to discriminate between the two forecasts in finite samples when only the proxy is available. In contrast, under exact robustness, the magnitude of the expected differences is identical for the true target and the proxy. Thus, our concept of exact robustness almost immediately implies that ECPA tests can be carried out as usual under measurement error, in particular allowing us to do a local power analysis. While our exact robustness is a stronger requirement than \citeauthor{Pat11}'s \citeyearpar{Pat11} ordering robustness, we show in Theorem~\ref{thm:1} that---surprisingly---the respective classes of loss functions coincide.
By doing so, we also refine the results of \citet{Pat11} along several dimensions; see Remarks~\ref{rem:1} and~\ref{rem:Patton} for details.

Another important aspect of our first main contribution is that---unlike \citet{Pat11} and \citet{Laurent2013}---we do not restrict attention to the mean functional from the outset. Thus, by narrowing down the class of functionals that can be evaluated robustly to the mean functional, we are able to show that (e.g.) the ranking of median forecasts is affected by noisy proxies. 
In Appendix~\ref{sec:NonRobustQuantiles}, we strengthen this result by showing that median (and more generally, quantile) forecasts cannot be evaluated robustly, even when the conditional unbiasedness assumption is replaced by \textit{any} other ``resemblance condition'' on the proxy. In other words, a robust evaluation of quantile forecasts requires the proxy and the true target to coincide, that is, it requires the absence of measurement error. However, the absolute error loss---pertaining to median forecasts---has regularly been used in comparing forecasts for mismeasured variables, such as inflation \citep{Han05,Mea21}, GDP growth \citep{RW09,BK14} and integrated variances \citep{HL05}. In each case, our results suggest that these comparisons should be interpreted with extreme caution, due to the non-robustness of the median.

Our second main contribution is to study the local power of ECPA tests using proxy variables and exactly robust loss functions. We demonstrate that power increases for more accurate proxies. The (infeasible) upper bound for the test power is obtained when evaluating forecasts with the most accurate proxy. In our case, this ``proxy'' is the---generally even ex post---latent target functional, i.e., the conditional mean of the target variable.
This supports the intuition that it is easier to discriminate between competing forecasts if the target is approximated more precisely.


Our simulations show that the asymptotic local power of the proxy-based ECPA test provides a good approximation in finite samples. We further demonstrate the dangers of using non-robust loss functions for comparing predictive accuracy with proxy variables. Specifically, size distortions may arise and, for certain alternatives, a loss of power occurs. These drawbacks, instead of getting less serious in larger samples, get more pronounced as the sample size increases.

We apply our proxy-based ECPA test to GDP growth rates in the US. GDP is a measure of the aggregate real output of the economy and, as such, is perhaps the most important macroeconomic indicator. However, GDP (and, hence, also GDP growth) cannot be measured exactly for various reasons. For instance, tax returns are incorporated into the national accounts only over time, leading to frequent revisions of GDP estimates (and thus different \textit{vintages}, i.e., series of GDP releases). Also, the US Bureau of Economic Analysis relies on economic census data collected only once every five years for computing GDP. Hence, GDP estimates are inherently based on some extrapolation, leading to error. Thus, for comparing forecasts, we can only use approximations of true GDP growth.

Several proxies for true GDP growth (denoted $\dGDP$) are available \citep{LSF08}. The arguably most popular proxy is the expenditure-side approximation $\dGDP_E$, followed by the income-side proxy $\dGDP_I$. 
Our third proxy, $\dGDP_+$ from \citet{Aea16} combines both of these information sources and can be regarded as a more precise proxy of latent GDP growth. 
Informed by our theory, we anticipate that using $\dGDP_+$ in ECPA tests leads to better discrimination between different forecasts. This is indeed what we find when comparing mean predictions from the Survey of Professional Forecasters (SPF) issued at different horizons. Naturally, we expect that $\dGDP$ forecasts for some time $t$ that were issued one quarter ago to be superior to those issued two or even four quarters ago. We confirm this and find that the evidence in favor of shorter horizon SPF forecasts is more convincing, the more precise the proxy. We obtain similar results for Greenbook forecasts, and also for more recent vintages of $\dGDP_E$, where later vintages typically provide more accurate proxies of true GDP growth.

The remainder of the paper proceeds as follows. 
In Section~\ref{Main Results}, we define \textit{exact/ordering robustness to measurement error} for loss functions, characterize these loss functions and show that exact and ordering robustness are equivalent. 
Then, we derive the local power of ECPA tests based on proxy variables and robust loss functions.
Section~\ref{Simulations} numerically illustrates the local power results in Monte Carlo experiments. Section~\ref{Empirical Application} applies our test to US GDP growth forecasts. Finally, Section~\ref{Conclusion} concludes. 
Appendix \ref{sec:NonRobustQuantiles} establishes a non-robustness result for quantile forecasts under any ``resemblance condition'' on the proxy, and the Appendices~\ref{Proofs} and~\ref{Technical Derivations} contain all proofs and some further technical derivations.

\section{Main Results}\label{Main Results}

\subsection{Characterizing Robust Loss Functions}\label{Characterizing Robust Loss Functions}

Let $(\Omega, \mathcal{A}, \p)$ be a probability space. 
If not stated otherwise, all (in-)equalities involving conditional expectations are tacitly assumed to hold $\p$-almost surely (a.s.) in the following.
Let $\mathcal{P}$ denote some class of distribution functions on $\mathbb{R}^{k}$ ($k\in\mathbb{N}$), which is specified later on.
For some time point $t\in\mathbb{N}$, we denote the random variable of interest (e.g., GDP growth) by $Y_t:\Omega\rightarrow\mathsf{O}\subset\mathbb{R}^{k}$, the time-$(t-1)$ information set by $\mathcal{F}_{t-1}$, and the conditional distribution of $Y_{t}$ given $\mathcal{F}_{t-1}$ by $F_t(\omega,\cdot)=\p\{Y_t\leq\cdot\mid\mathcal{F}_{t-1}\}(\omega)$, where we assume that $F_t(\omega,\cdot)\in\mathcal{P}$ for all $\omega\in\Omega$. 
We denote by $F_t(\cdot)$ the random variable defined by the mapping $\omega \mapsto F_t(\omega,\cdot)$.

\begin{rem}
The theory of this article can be extended to $\tau$-step ahead forecasts for $\tau \ge 2$ in a straightforward fashion by considering the information set $\mathcal{F}_{t-\tau}$ instead of $\mathcal{F}_{t-1}$. Since we leave $\mathcal{F}_{t-1}$ unspecified in the following, this can be seen by letting $\mathcal{F}_{t-1}$ contain only information available at time $(t-\tau)$.
We deliberately do not make this explicit in the notation in order to keep the exposition as simple as possible.
\end{rem}

The aim in a forecasting situation is to predict some target functional $T:\mathcal{P}\rightarrow\mathsf{A}\subset\mathbb{R}^{k}$ of $F_t$. We denote the target functional by $x_t^\ast:\Omega\rightarrow\mathsf{A}$, $\omega\mapsto T(F_t(\omega,\cdot))$. For the leading case $k=1$, this may be the conditional mean of GDP growth, $x_t^\ast=\E_{t-1}[Y_t]$, where we write $\E_{t-1}[\,\cdot\,]=\E[\ \cdot\mid\mathcal{F}_{t-1}]$ for short. We assume that there exists a \textit{strictly $\mathcal{P}$-consistent} (or simply \textit{strictly consistent}) loss function $L:\mathsf{O}\times\mathsf{A}\rightarrow\mathbb{R}$ for the functional $T$, that is
\begin{equation}\label{eq:(SC)}
	\E_{Y \sim F}\big[L(Y,T(F))\big]\leq\E_{Y \sim F}\big[L(Y,x)\big], 
\end{equation}
for all $F\in\mathcal{P}$ and all $x\in\mathsf{A}$, and equality in \eqref{eq:(SC)} implies $x=T(F)$; see \citet{Gne11}.
Here, the notation $\E_{Y \sim F}[\,\cdot\,]$ denotes the expectation with respect to $Y$ with distribution $F$.
A functional for which a strictly consistent loss function exists is termed \textit{elicitable}. 
Assuming $T$ to be elicitable is not restrictive in our context, because when no strictly consistent scoring function exists, forecasts cannot be compared validly \citep{Gne11}.
Denote by $x_{it}:\Omega\rightarrow\mathsf{A}$ ($i=1,2$) the two competing, $\mathcal{F}_{t-1}$-measurable forecasts of $x_t^\ast$. The forecast loss difference
\[
	d(Y_t, x_{1t}, x_{2t}) = L(Y_t, x_{1t}) - L(Y_t, x_{2t})
\]
measures the relative performance of $x_{1t}$ and $x_{2t}$. Since the loss function is negatively oriented, a negative (positive) loss difference favors $x_{1t}$ ($x_{2t}$).

As pointed out in the Motivation, the true $Y_t$ may often not be available for comparing forecasts due to measurement error. Instead, one has to rely on a proxy $\widehat{Y}_t:\Omega\rightarrow\mathsf{O}$ for $Y_t$ and use $d(\widehat{Y}_t, x_{1t}, x_{2t})$. Of course, the proxy has to bear some resemblance to the target. To ensure this, we make the assumption that $\widehat{Y}_t$ is \textit{conditionally unbiased} for $Y_t$, i.e., that $\E_{t-1}[\widehat{Y}_t]\overset{\text{a.s.}}{=}\E_{t-1}[Y_t]$.
This also implies that (mean) differences between $Y_t$ and $\widehat{Y}_t$ cannot be predicted from information available at time $(t-1)$. We consider this to be a natural assumption for a proxy variable, since otherwise one could predict the average measurement error $\widehat{Y}_t-Y_t$ based on $\mathcal{F}_{t-1}$. We refer to the empirical application for a modeling framework of GDP growth rates, where the conditional unbiasedness assumption is satisfied for $Y_t=\dGDP_{t}$ and $\widehat{Y}_t\in\big\{\dGDP_{E,t}, \dGDP_{I,t}\big\}$; see in particular \eqref{eq:GDP}.


We denote the conditional distribution function of $\widehat{Y}_t\mid\mathcal{F}_{t-1}$ by $\widehat{F}_{t}(\omega,\cdot)$ for all $\omega \in \Omega$, and the corresponding random variable by $\widehat{F}_{t}(\cdot)$. The following definition ensures that the expected forecast loss differences are the same for $Y_t$ and a conditionally unbiased proxy $\widehat{Y}_t$.

\begin{defn}\label{defn:1}
$L(\cdot,\cdot)$ is \textit{\underline{exactly} robust to measurement error} (or simply: \textit{exactly robust}) with respect to $\mathcal{P}$, if
\[
	\E_{t-1}[d(Y_t, x_{1t}, x_{2t})]=\E_{t-1}[d(\widehat{Y}_t, x_{1t}, x_{2t})]\qquad\text{a.s.}
\]
for all $\mathcal{F}_{t-1}$-measurable forecasts $x_{1t}$ and $x_{2t}$, and all $Y_t$ and all conditionally unbiased proxies $\widehat{Y}_t$ with $F_{t}(\omega,\cdot)\in\mathcal{P}$ and $\widehat{F}_{t}(\omega,\cdot)\in\mathcal{P}$ for all $\omega\in\Omega$.
\end{defn}

Exact robustness to measurement error suggests that we can learn as much about the relative merits of $x_{1t}$ and $x_{2t}$ by observing $d(\widehat{Y}_t, x_{1t}, x_{2t})$ instead of $d(Y_t, x_{1t}, x_{2t})$. However, this is only correct in expectation, or equivalently, asymptotically under a suitable law of large numbers. In finite samples, this is unfortunately not true, since $d(\widehat{Y}_t, x_{1t}, x_{2t})$ may have a larger (or smaller) variance than $d(Y_t, x_{1t}, x_{2t})$.

	The more restricted the class $\mathcal{P}$ in Definition~\ref{defn:1}, the richer the class of exactly robust loss functions. An extreme case arises if $\mathcal{P}=\{F:\mathbb{R}^{k
	}\rightarrow[0,1] \mid F(x)=I_{\{x\geq o\}}\text{ for some }o\in\mathsf{O}\}$ is the class of degenerate distributions. (The inequality $x\geq o$ is to be understood component-wisely if the quantities involved are vector-valued.) Then, by (a.s.) constancy of $Y_t$ and $\widehat{Y}_t$, we necessarily have that $Y_t \overset{\text{a.s.}}{=} \widehat{Y}_t$ by conditional unbiasedness. The latter implies that \textit{any} loss function is exactly robust with respect to this rather restricted class $\mathcal{P}$.

\begin{rem}\label{rem:0}
Consider univariate log-returns $r_t$ on some speculative asset and assume that $\E[r_t\mid\mathcal{F}_{t-1}]=0$, as is common for financial data. In this setting, forecasts of the (latent) conditional variance $x_t^\ast=\E[r_t^2\mid\mathcal{F}_{t-1}] = \Var(r_t\mid\mathcal{F}_{t-1})$ are essential for risk management purposes \citep{Aea13}. Thus, in our setting, the target variable is $Y_t=r_t^2$ and $T$ is the mean functional. To compare two volatility forecasts $x_{1t}$ and $x_{2t}$, the natural choice is then $Y_t=r_t^2$. However, in the---different, but related---context of evaluating volatility forecasts of GARCH models, \citet{AB98} advocate the use of less noisy high-frequency proxies $\widehat{Y}_t$ of the conditional variance (such as realized volatility or the range) satisfying $\E_{t-1}[\widehat{Y}_t]=\E_{t-1}[Y_t]=x_t^\ast$. This example shows that (other than our notation for $Y_t$ and $\widehat{Y}_t$ suggests) $\widehat{Y}_t$ does not necessarily have to be regarded as coming ``close'' to $Y_t$. Instead, we may sometimes interpret $Y_t$ \textit{and} $\widehat{Y}_t$ as providing two different conditionally unbiased estimates of the target functional $x_t^\ast$.
\end{rem}



\begin{rem}\label{rem:1}
In the framework of Remark~\ref{rem:0}, \citet{Pat11} investigates conditions under which loss functions produce consistent \textit{relative rankings} in the sense that
\begin{equation*}
	\E[L(Y_t, x_{1t})]\lesseqgtr\E[L(Y_t, x_{2t})] \quad\Longleftrightarrow\quad \E[L(\widehat{Y}_t, x_{1t})]\lesseqgtr\E[L(\widehat{Y}_t, x_{2t})]
\end{equation*}
for any (volatility) proxy satisfying $\E_{t-1}[\widehat{Y}_t]=\E_{t-1}[Y_t](=x_t^\ast)$. 
This property is conceptually different from \textit{exact} robustness to measurement error in that, first, only the ranking of forecasts is concerned and, second, unconditional means are considered. \citet{Laurent2013} generalize \citeauthor{Pat11}'s \citeyearpar{Pat11} results by considering multivariate $r_t$, where $Y_t=\vech(r_{t}r_t^\prime)$ is the target variable and the target functional $x_t^\ast=\E[Y_t\mid\mathcal{F}_{t-1}]$ is the ($\vech$-transformed) conditional variance-covariance matrix of the returns. Here, $\vech(\cdot)$ stacks the lower triangular part of a matrix into a vector.
\end{rem}

Despite the conceptual differences it will be insightful to transfer the idea of loss functions that produce consistent relative rankings to our framework. To do so, we generalize the definition of \citet{Pat11} as follows:

\begin{defn}\label{defn:2}
$L(\cdot,\cdot)$ is \textit{\underline{ordering} robust to measurement error} (or simply: \textit{ordering robust}) with respect to $\mathcal{P}$, if
\begin{equation*}
	\E_{t-1}[d(Y_t, x_{1t}, x_{2t})]\lesseqgtr0 \quad \text{a.s.}
	\quad\Longleftrightarrow\quad 
	\E_{t-1}[d(\widehat{Y}_t, x_{1t}, x_{2t})]\lesseqgtr0 \quad \text{a.s.}
\end{equation*}
for all $\mathcal{F}_{t-1}$-measurable forecasts $x_{1t}$ and $x_{2t}$, and all $Y_t$ and all conditionally unbiased proxies $\widehat{Y}_t$ with $F_{t}(\omega,\cdot)\in\mathcal{P}$ and $\widehat{F}_{t}(\omega,\cdot)\in\mathcal{P}$ for all $\omega\in\Omega$.
\end{defn}


We argue that exact robustness is a more useful concept in practice than the weaker ordering robustness. To see why, assume that $\E_{t-1}[d(Y_t, x_{1t}, x_{2t})]=\nu_t$ for some large $\nu_t>0$, such that $x_{2t}$ is a much better forecast. Under exact robustness, we then have $\E_{t-1}[d(\widehat{Y}_t, x_{1t}, x_{2t})]=\nu_t$, such that $x_{2t}$ is again clearly superior when judged using the proxy $\widehat{Y}_t$. However, under ordering robustness, we may merely have $\E_{t-1}[d(\widehat{Y}_t, x_{1t}, x_{2t})]=\widetilde{\nu}_t$ for some small $\widetilde{\nu}_t>0$.
 Hence, in finite samples, it may be much harder to identify $x_{2t}$ as the better forecast when only the $\widehat{Y}_t$ are available. In other words, two forecasts may be easier to separate when making an exactly robust comparison instead of an ordering robust comparison. Thus, when the forecasts and the proxy are taken as given in applied work, we advocate the use of exactly robust loss functions. However, this recommendation is based solely on conceptual considerations, because---foreshadowing one of the main results of Theorem~\ref{thm:1}---the classes of exactly robust and ordering robust loss functions coincide.

Before we can state Theorem~\ref{thm:1}, recall the concept of a \textit{subgradient} from convex analysis. To do so, we let $\langle\cdot,\cdot\rangle$ denote the standard scalar product in $\mathbb{R}^{k}$. Then, the subgradient at value $x \in \mathsf{A}$ of some convex function $\phi:\mathsf{A}\rightarrow\mathbb{R}$, denoted $\D\phi(x)$, is any vector $s\in\mathsf{A}$ satisfying $\phi(y)\geq\phi(x)+\langle s, y-x\rangle$ for all $y\in\mathsf{A}$ \citep[Def.~1.2.1]{HL01}. Recall from \citet[Sec.~D]{HL01} that such a vector always exists.

\begin{thm}\label{thm:1}
Let $\mathcal{P}$ be a convex set of distribution functions. Assume that $T:\mathcal{P}\rightarrow\mathsf{A}$ is surjective, $\mathsf{A}$ is convex, and $L(\cdot,\cdot)$ is strictly $\mathcal{P}$-consistent for $T$. Then, the following are equivalent:
\renewcommand{\labelenumi}{(\alph{enumi})}
\begin{enumerate}
\item\label{item:form} $L(\cdot,\cdot)$ is of the form
\begin{equation}\label{eq:L char}
	L(Y, x)=\phi(x) + \langle \D\phi(x), Y-x \rangle + a(Y),
\end{equation}
where $\phi:\mathsf{A}\rightarrow\mathbb{R}$ is strictly convex with subgradient $\D\phi(\cdot)$, and $a:\mathsf{O}\rightarrow\mathbb{R}$ is integrable with respect to all $F\in\mathcal{P}$;
\item $L(\cdot,\cdot)$ is exactly robust with respect to $\mathcal{P}$;
\item $L(\cdot,\cdot)$ is ordering robust with respect to $\mathcal{P}$;
\item For all $Y_t$ with $F_t(\omega,\cdot)\in\mathcal{P}$ for all $\omega\in\Omega$, it holds that $x_t^\ast=T(F_{t}(\cdot))\overset{\text{a.s.}}{=}\E_{t-1}[Y_t]$.
\end{enumerate}
\end{thm}

Theorem~\ref{thm:1} offers three main insights. First, it characterizes the exactly robust loss functions, while making \textit{no} assumption on the functional to be evaluated in advance. Exact robustness is essential to theoretically study ECPA tests under the alternative, where there is a difference in predictive ability, i.e.,~$\E_{t-1}[d(Y_t,x_{1t}, x_{2t})]\neq0$. In the absence of the equivalence of (b) and (c) established in Theorem~\ref{thm:1}, mere ordering robustness would not be sufficient to do so, because the magnitude of the deviation from zero in $\E_{t-1}[d(\widehat{Y}_t,x_{1t}, x_{2t})]\neq0$ may be changed. Thus, a second contribution of Theorem~\ref{thm:1} is the equivalence of exact and ordering robustness, and we merely speak of robust loss functions in the following when no confusion can arise. The third insight is that if, in practice, there is measurement error in the target variable, then \textit{only} conditional mean forecasts can be evaluated robustly. For other functionals, such as the median, the forecast ranking is affected by the use of proxies. In particular, many commonly used loss functions, such as absolute error (AE) loss $L(Y,x)=|Y-x|$, are not robust. Nonetheless, the AE loss has been used in comparing forecasts for mismeasured variables, such as inflation \citep{Han05,Mea21}, GDP growth \citep{RW09,BK14} and integrated variances \citep{HL05}. Thus, Theorem~\ref{thm:1} casts doubt on the rankings obtained by these comparisons.

Theorem~\ref{thm:1} crucially depends on the conditional unbiasedness condition, $\E_{t-1} [ \widehat Y_t ] \overset{\text{a.s.}}{=} \E_{t-1} [ Y_t ]$.
This raises the question if forecasts for other target functionals than the mean can be compared robustly under alternative ``resemblance conditions'' on $Y_t$ and $\widehat Y_t$.
In Appendix \ref{sec:NonRobustQuantiles}, we show that conditional quantile forecasts cannot be evaluated (exactly) robustly, unless one imposes the degenerate resemblance condition that the distributions of $Y_t$ and $\widehat Y_t$ coincide.
This reinforces our interpretation of the mean being the only target functional that allows for robust evaluation.
Thus, we have the surprising result that while the mean is more robust than the median in forecast evaluation, the opposite is well-known to hold in classical estimation theory.


\begin{rem}\label{rem:Patton}
The equivalence of (a) and (c) in Theorem~\ref{thm:1} may be viewed as a generalization of Proposition~1 in \citet{Pat11} for $k=1$ and of Proposition~2 in \citet{Laurent2013} for $k\in\mathbb{N}$. Since \citet{Laurent2013} use very similar regularity conditions as \citet{Pat11}, we only highlight the main improvements on the latter.
	First, while Patton's characterization builds on the \textit{unconditional} expectation, we consider \textit{conditional} expectations, which is crucial for tests of equal \textit{conditional} predictive ability \citep{GW06} and tests of superior \textit{conditional} predictive ability \citep{LLQ21+}.
	Second, similar to the property of strict consistency for loss functions \citep{Gne11}, robustness should be considered with respect to a specified class $\mathcal{P}$ of distributions. While we merely require $\mathcal{P}$ to be convex, \citet[A2]{Pat11} restricts attention to absolutely continuous distributions.
	Third, \cite{Pat11} only considers  to continuously differentiable losses, which ignores important classes of loss functions, such as the generalized piecewise linear (GPL) losses. In contrast, 
	our Theorem~\ref{thm:1} dispenses with \textit{any} regularity conditions on the class of possible loss functions. Thus, our class of loss functions in \eqref{eq:L char} is broader than his \textit{and} we do not rule out non-differentiable losses from the outset. Fourth, by considering $x_t^\ast=\E[r_t^2\mid\mathcal{F}_{t-1}]$ (cf.~Remark~\ref{rem:0}),	\citet[A1 \& A4]{Pat11} specifies $T$ to be the mean functional as an \textit{assumption}, such that implicitly only Bregman loss functions are considered at the outset.	In contrast, we do not specify $T$ in advance, allowing us to show the non-robustness of \textit{any} functional apart from the mean. We view this final improvement as the most important one due to its practical implication: when there is measurement error, one can only evaluate mean forecasts robustly. The forecast ranking of any other other functional (e.g., the median) will be affected.
\end{rem}

\begin{rem}\label{rem:Estimation}
	The classification of loss functions in Theorem~\ref{thm:1} can be employed  (by invoking Theorem~2.5 in \citet{DFZ2020}) for the objective functions in M-\textit{estimation} of semiparametric models, when only a proxy $\widehat Y_t$ of the response variable, $Y_t$, is observable.
	Theorem \ref{thm:1} then implies that for conditionally unbiased proxies, the M-estimator is consistent for, and only for, conditional \textit{mean} models; in contrast to, e.g., conditional \textit{quantile} models.
\end{rem}

\subsection{Testing Equal Predictive Accuracy}\label{Testing Predictive Accuracy}

Section~\ref{Characterizing Robust Loss Functions} shows that for loss functions of the form \eqref{eq:L char}, the expected loss differences are unchanged when using a conditionally unbiased proxy $\widehat{Y}_t$. Here, we consider the implications of Theorem \ref{thm:1} for statistical tests of ECPA, and show that the robustness property leads to valid tests (in the sense that size is kept) whose power increases for more accurate proxies. 

To that end, we outline the framework of \textit{conditional} predictive ability testing pioneered by \citet{GW06}, which extends the classical predictive ability tests of \citet{DM95}. Recall that interest in ECPA tests centers on the null hypothesis
\[
	H_0\ :\quad \E_{t-1}[d(Y_t, x_{1t}, x_{2t})]\overset{\text{a.s.}}{=}0\qquad\text{for all }t=1,2,\ldots. 
\]
To test this \textit{conditional} moment condition based on a finite sample of length $n$, \citet{GW06} propose to test the implication of $H_0$ that $\E[h_{t-1}d(Y_t, x_{1t}, x_{2t})]=0$ for a $\mathcal{F}_{t-1}$-measurable test function $h_{t-1}$, taking values in $\mathbb{R}^{q}$. They do so using the Wald-type test statistic $T_n=n\overline{Z}_n^{\prime}\widetilde{\Omega}_n^{-1}\overline{Z}_n$, where $Z_t=h_{t-1}d(Y_t, x_{1t}, x_{2t})$, $\overline{Z}_n=1/n\sum_{t=1}^{n}Z_t$, and $\widetilde{\Omega}_n$ is an invertible and consistent estimator of $\Var(\sqrt{n} \, \overline{Z}_n)$. Under weak regularity conditions, \citet[Theorem~1]{GW06} show that $T_n$ is asymptotically $\chi^2_{q}$-distributed under $H_0$. 

However, when $Y_t$ is not observed, $T_n$ cannot be computed and we have to rely on the feasible test statistic
\[
	\widehat{T}_n=n\overline{\widehat{Z}}_n^{\prime}\widehat{\Omega}_n^{-1}\overline{\widehat{Z}}_n,
\]
where $\widehat{Z}_{t}=h_{t-1}d(\widehat{Y}_t, x_{1t}, x_{2t})$, $\overline{\widehat{Z}}_n=1/n\sum_{t=1}^{n}\widehat{Z}_t$, and $\widehat{\Omega}_n$ is an invertible and consistent estimator of $\Omega_n=\Var(\sqrt{n}\overline{\widehat{Z}}_n)$. Arguing as before, $\widehat{T}_n$ follows a $\chi^2_{q}$-distribution asymptotically under the proxy hypothesis
\[
	\widehat{H}_0\ :\ \E_{t-1}[d(\widehat{Y}_t, x_{1t}, x_{2t})]\overset{\text{a.s.}}{=}0\qquad\text{for all }t=1,2,\ldots.
\]
Hence, in the presence of measurement error, $\widehat{T}_n$ can validly test $H_0$ if $\E_{t-1}[d(Y_t, x_{1t}, x_{2t})]\overset{\text{a.s.}}{=}\E_{t-1}[d(\widehat{Y}_t, x_{1t}, x_{2t})]$, i.e., if $H_0$ and $\widehat{H}_0$ are equivalent, which is obviously implied by our exact robustness property.

Under exact robustness to measurement error, $H_0$ implies that $\E[d(Y_t,x_{1t}, x_{2t})] = 0 = \E[d(\widehat{Y}_t,x_{1t}, x_{2t})]$ by the law of iterated expectations. Thus, when $L(\cdot,\cdot)$ is of the form given in Theorem~\ref{thm:1}, EPA tests of the hypothesis $\E[d(Y_t,x_{1t}, x_{2t})] = 0$ can also be validly carried out using a conditionally unbiased proxy $\widehat{Y}_t$ of $Y_t$.

\begin{rem}
	\label{rem:LP18}
	For HF proxies in finance, \cite{LP18} establish that equality of the expected loss differences is not strictly necessary for the validity of equal predictive ability tests under measurement error. 
	Instead, it suffices that the expected loss differences converge at rate $o(\sqrt{n})$, which they call the ``convergence-of-hypotheses'' condition.
	However, verification of the latter often requires non-trivial primitive conditions. E.g., when comparing volatility forecasts with high-frequency proxies, the sample size of intraday returns for computing volatility proxies must diverge faster than the number of out-of-sample volatility forecasts. In macroeconomic applications, where the sampling frequency cannot be arbitrarily increased, the ``convergence-of-hypotheses'' condition is not applicable.
	One conclusion of our Theorem~\ref{thm:1} is that the convergence-of-hypotheses condition is not required (because it holds trivially) for mean forecasts when conditionally unbiased proxies are used.
\end{rem}	


Next, we derive the limit of $\widehat{T}_n$ under the local alternative
\begin{align}
	\label{eq:HAloc}
	H_{a,\loc}\ :\ \E[h_{t-1}d(Y_t, x_{1t}, x_{2t})]=\frac{\delta}{\sqrt{n}}\quad\text{for all }t=1,2,\ldots,
\end{align}
where $\delta\in\mathbb{R}^{q}$. The magnitude of the local alternative, $\delta/\sqrt{n}$, converges to the null hypothetical value of zero as $n\to\infty$, thus making it harder for our test to reject $H_0$ for increasing $n$. Since $\E[Z_t]=\delta/\sqrt{n}$ under $H_{a,\loc}$, $Z_t$ depends on $n$. We reflect this in our notation by writing $Z_t=Z_{n,t}=(Z_{n,t}^{(1)},\ldots,Z_{n,t}^{(q)})^\prime$. Note that for $\delta=0$, $H_{a,\loc}$ is strictly speaking not an alternative as it reduces to the null. However, the method of proof for deriving the asymptotic limit of $\widehat{T}_n$ is the same for all $\delta\in\mathbb{R}^{q}$. Hence, we leave $\delta$ unrestricted in $H_{a,\loc}$. Values of $\delta$ with larger norm correspond to local alternatives with larger magnitudes. Note that under fixed alternatives, ECPA tests are consistent, i.e., power converges to one \textit{no matter} which conditionally unbiased proxy is used. Thus, only by considering \textit{local} alternatives,  we are able to derive analytical results on the power of tests that use different proxies.

If the underlying loss function is robust, then it follows under $H_{a,\loc}$ by the law of iterated expectations (LIE) that
\begin{equation}
	\label{eq:(p.8)}
	\frac{\delta}{\sqrt{n}}=\E[h_{t-1}d(Y_t, x_{1t}, x_{2t})] 
	= \E[h_{t-1}\E\{d(Y_t, x_{1t}, x_{2t})\mid\mathcal{F}_{t-1}\}] 
	= \E[h_{t-1}d(\widehat{Y}_t, x_{1t}, x_{2t})]
\end{equation}
for any conditionally unbiased proxy $\widehat{Y}_t$. The property in \eqref{eq:(p.8)}, implied by exact robustness, is essential for our local power results, because it allows to test $\E[h_{t-1}d(Y_t, x_{1t}, x_{2t})]=\delta/\sqrt{n}$ via the $\widehat{Z}_{n,t}=h_{t-1}d(\widehat{Y}_t, x_{1t}, x_{2t})$. Thus, one key benefit of exact robustness is that it allows to compare ``how difficult'' it is to assess the relative merits of two forecasts, when only an imperfect proxy $\widehat{Y}_{t}$ is available. Note that even though for different proxies $\widehat{Y}_t$ the deviations under \eqref{eq:(p.8)} are equal, better proxies may still give less variable $\widehat{Z}_{n,t}$, such that departures from $H_0$ of the form in \eqref{eq:(p.8)} may be detected more easily. To investigate this, we derive the local power of the test based on $\widehat{T}_n$.

To do so, we make the following assumptions, which are similar to those of \citet{GW06}. As only some proxy variable $\widehat{Y}_t$ is available, we specify these assumptions in terms of $\widehat{Z}_{n,t}=h_{t-1}d(\widehat{Y}_t, x_{1t}, x_{2t})$. Nonetheless, the specific (conditionally unbiased) choices $\widehat{Y}_t=Y_t$ and $\widehat{Y}_t=x_t^\ast$ are also allowed. Let $\widehat{W}_t=(\widehat{Y}_t^\prime,X_t^\prime)^\prime$, where the $\mathcal{F}_{t-1}$-measurable $X_t$ is $\mathbb{R}^{s}$-valued and contains all predictors that the forecasts $x_{1t}$ and $x_{2t}$ are based on. For instance, $X_t$ may contain lagged $\widehat{Y}_{t}$'s.

\begin{enumerate}
	\item[D1:] $\big\{(\widehat{W}_t^\prime,h_t^\prime)^\prime\big\}$ is $\alpha$-mixing of size $-2r/(r-2)$ or $\phi$-mixing of size $-r/(r-1)$.
	\item[D2:] $\E|\widehat{Z}_{n,t}^{(i)}|^{2r}\leq\Delta_{Z}<\infty$ for $i=1,\ldots,q$ and $r$ from D1.
	\item[D3:] $\Omega_n=\Var\big({n}^{-1/2}\sum_{t=1}^{n}\widehat{Z}_{n,t}\big)\longrightarrow\Omega$, as $n\to\infty$, where $\Omega$ is positive definite.
	\item[D4:] The forecasts $x_{1t}=f_{n,t}^{(1)}(X_t,\ldots,X_{t-m_1+1})$ and $x_{2t}=f_{n,t}^{(2)}(X_t,\ldots,X_{t-m_2+1})$ are measurable functions of a finite number of predictors, where $m=\max(m_1,m_2)\leq\overline{m}<\infty$.
\end{enumerate}

Our conditions D1--D4 are standard in the literature, and are essentially those of \citet[Theorem~3]{GW06}. Assumption~D4, that forecasts are based on a finite number of lags of the predictor variables, is a technical convenience used to expedite the proof of Theorem~\ref{thm:2}. Nonetheless, this assumption accommodates many different estimation schemes, where $m_1$ and $m_2$ may be deterministic or data-driven (and, hence, stochastic). We refer to \citet{GW06} for more detail. Note that in D4 we suppress the dependence of the forecasts on $n$ for notational brevity.

\begin{thm}\label{thm:2}
Let Assumptions \textup{D1}--\textup{D4} hold. Then, if $\widehat{\Omega}_n-\Omega_n=o_{\p}(1)$, it holds for any exactly robust loss function under $H_{a,\loc}$ that
\[
	\widehat{T}_n\overset{d}{\longrightarrow}\chi_{q}^2(\delta^\prime\Omega^{-1}\delta),\qquad\text{as }n\to\infty,
\]
where $\chi_{q}^2(c)$ denotes a $\chi_q^2$-distribution with non-centrality parameter $c\in\mathbb{R}$.
\end{thm}

In the special case when $Y_t=\widehat{Y}_t$, Theorem~\ref{thm:2} is a refinement of the consistency under fixed alternatives established by \citet[Theorem~2]{GW06}. For $\delta=0$, Theorem~\ref{thm:2} shows that $\widehat{T}_n$ has a standard $\chi^2$-limit under $H_0$. Moreover, it implies that the asymptotic local power (ALP), i.e., the asymptotic probability of rejecting $H_0$ under $H_{a,\loc}$, is
\begin{equation}\label{eq:ALP}
	\p\big\{\chi_q^{2}(\delta^\prime\Omega^{-1}\delta)>\chi_{q,1-\tau}^{2}(0)\big\},
\end{equation}
where $\chi_{q,1-\tau}^{2}(0)$ is the $(1-\tau)$-quantile of the $\chi_{q}^{2}(0)$-distribution and $\tau \in (0,1)$ is the significance level of the test.
This shows that the ALP depends on the magnitude of the local alternative (via $\delta$) \textit{and} on the proxy via the limit $\Omega$ of $\Omega_n=\Var\big({n}^{-1/2} \sum_{t=1}^{n}\widehat{Z}_{n,t}\big)$. 
Thus, more precise proxies that give less variable $\widehat{Z}_{n,t}$ lead to tests with higher local power. We shed more light on this in the next subsection.

To make the test based on $\widehat{T}_n$ operational, we need a consistent estimator $\widehat{\Omega}_n$ of $\Omega_n$. 
To that end, we consider a heteroskedasticity and autocorrelation consistent (HAC) estimator \citep{Newey/West:87a}
\begin{equation}\label{eq:Omega def}
	\widehat{\Omega}_n=\frac{1}{n}\sum_{t=1}^{n}\widehat{Z}_{n,t}\widehat{Z}_{n,t}^\prime + \frac{1}{n}\sum_{h=1}^{m_n}w_{n,h}\sum_{t=h+1}^{n}\Big(\widehat{Z}_{n,t}\widehat{Z}_{n,t-h}^\prime + \widehat{Z}_{n,t-h}\widehat{Z}_{n,t}^\prime\Big),
\end{equation}
where $m_n$ is a sequence of integers, and $w_{n,h}$ is a scalar triangular array of weights. We restrict $m_n$ and $w_{n,h}$ as follows:
\begin{itemize}
	\item[D5:] The sequence of integers $m_n$ satisfies $m_n\rightarrow\infty$ and $m_n=o(n^{1/4})$, as $n\to\infty$.
	\item[D6:] It holds that $|w_{n,h}|\leq\Delta_{w}<\infty$ for all $n\in\mathbb{N}$ and $h\in\{1,\ldots,m_n\}$, and $w_{n,h}\rightarrow1$, as $n\to\infty$, for all $h=1,\ldots,m_n$.
\end{itemize}

\begin{prop}\label{prop:LRV}
Let Assumptions \textup{D1}--\textup{D6} hold. Then, it holds under $H_{a,\loc}$ that $\widehat{\Omega}_n-\Omega_n=o_{\p}(1)$, as $n\to\infty$.
\end{prop}

The estimator $\widehat{\Omega}_n$ is the omnibus choice. It works for EPA tests (where $H_0$ reduces to $\E[d(Y_t,x_{1t}, x_{2t})] = 0$ by the LIE) but also for ECPA tests. For the latter tests, simpler estimators may be used, because---under $H_0$---the sequences $\{h_{t-1}d(\widehat{Y}_t,x_{1t}, x_{2t}),\mathcal{F}_t\}$ are martingale differences and thus, uncorrelated. This implies that $\Omega_n$ simplifies to $\Omega_n=1/n \sum_{t=1}^{n}\E[\widehat{Z}_{n,t}\widehat{Z}_{n,t}^\prime]$, rendering $\widehat{\Omega}_n= 1/n \sum_{t=1}^{n}\widehat{Z}_{n,t}\widehat{Z}_{n,t}^\prime$ the estimator of choice. However, even for EPA tests of one-step-ahead forecasts considered here, \citet{DM95} recommend to use $m_n=0$ in \eqref{eq:Omega def} (with an empty sum defined to be zero). Thus, using $\widehat{\Omega}_n=1/n\sum_{t=1}^{n}\widehat{Z}_{n,t}\widehat{Z}_{n,t}^\prime$ for both EPA and ECPA tests seems reasonable, and we opt for this choice in the numerical experiments in Section \ref{Simulations}.

\subsection{Finding Optimal Proxies}
\label{sec:Proxies}

Under the local alternative $H_{a,\loc}$, Theorem~\ref{thm:2} shows that the test's asymptotic local power is maximized by choosing a proxy $\widehat Y_t$ that minimizes $\Omega$.
We discuss such choices in the following by utilizing the linearity of the Bregman loss functions in $Y_t$. Specifically, we have

\begin{prop}
	\label{prop:VarYt}
	Suppose that the assumptions of Theorem~\ref{thm:1} hold and that any of (a)--(d) are in force. Then, it holds that
	\begin{align}
	\Omega_n 
	&= \frac{1}{n} \sum_{t=1}^n \Var \Big( h_{t-1} d(x_t^\ast, x_{1t}, x_{2t}) \Big) 
	+ \frac{1}{n} \sum_{t=1}^n \E \Big[  h_{t-1}  b_{t-1}^\prime  \Var_{t-1} \big( \widehat Y_t \big) b_{t-1} h_{t-1}^\prime\Big] \notag\\
	&\qquad+ \frac{2}{n} \sum_{s < t} \Cov \Big( h_{s-1} d(\widehat Y_s, x_{1s}, x_{2s}) ,\; h_{t-1} d(x_t^\ast, x_{1t}, x_{2t}) \Big),\label{eq:help Omega}
	\end{align}
	where $b_{t-1}=\D\phi(x_{1t})-\D\phi(x_{2t})$. Furthermore, if the covariance terms vanish asymptotically, 
	\begin{align}
	\label{eqn:ProxyVariance}
	\Omega_n 
	= \frac{1}{n} \sum_{t=1}^n \Var \Big( h_{t-1} d(x_t^\ast, x_{1t}, x_{2t}) \Big) 
	+ \frac{1}{n} \sum_{t=1}^n \E \Big[  h_{t-1}   b_{t-1}^\prime  \Var_{t-1} \big( \widehat Y_t \big)b_{t-1} h_{t-1}^\prime \Big]+o(1),
	\end{align}
	with an asymptotic lower bound of $\Omega^\ast=\lim_{n\to\infty}1/n \sum_{t=1}^n \Var \left( h_{t-1} d(x_t^\ast, x_{1t}, x_{2t}) \right)$ (in the sense that $\Omega-\Omega^{\ast}$ is positive semi-definite), which is attained if and only if $\widehat Y_t = x_t^\ast$ a.s. 
\end{prop}

The covariance terms in \eqref{eq:help Omega} vanish (and, hence, \eqref{eqn:ProxyVariance} holds) if the $\big\{ h_{t-1} d(\widehat{Y}_t, x_{1t}, x_{2t}) \big\}$ are serially uncorrelated, which holds under $H_0$. But \eqref{eqn:ProxyVariance} may also hold under $H_{a,\loc}$ as the derivation of $\Omega$ in Proposition~\ref{prop:Omega} of Appendix~\ref{Technical Derivations} shows. If \eqref{eqn:ProxyVariance} holds, the generally infeasible choice of $\widehat Y_t = x_t^\ast$ gives an upper bound for the local test power as then, $\Var_{t-1}\big( \widehat Y_t\big) \overset{\text{a.s.}}{=} 0$, such that the second (positive semi-definite) term on the right-hand side of \eqref{eqn:ProxyVariance} vanishes. This is very intuitive: It is easiest to assess the relative merits of two forecasts $x_{1t}$ and $x_{2t}$ if they can be compared against the target $x_t^{\ast}$ itself. This supports the intuition that it is easier to distinguish between two forecasts if the target functional is approximated more precisely. We stress once again that this result relies on the newly introduced exact robustness concept. While comparing forecasts with $x_t^{\ast}$ itself is best from a theoretical point of view, we are not aware of a practical situation where $x_t^\ast$ is observable, even ex post at time $t$.

Since the test power decreases with increasing (in terms of the Loewner order) average variances, $\Var_{t-1} \big( \widehat Y_t \big)$, the proxy $ \widehat Y_t$ with smallest possible \textit{conditional} variance should be employed in practice. 
However, these \textit{conditional} variances are generally unknown in practice.
Then, one often has to resort to domain-specific knowledge for choosing the best proxy.
We now discuss this exemplarily in the univariate case ($k=1$) for the volatility forecasting example of Remark~\ref{rem:0}, and for our macroeconomic application in Section~\ref{Empirical Application}.




In the macroeconomic application, $Y_t$ denotes true GDP growth.
As discussed in the Motivation, true GDP growth cannot be observed, even ex post. 
However, several different proxies $\widehat Y_t$ are available.
Suppose, as is plausible, that there is some additive measurement error, such that $\widehat Y_t = Y_t + \widehat{\varepsilon}_t$; see, e.g., \citet{FRW05} for such a modeling approach. Here, the estimation error $\widehat{\varepsilon}_t$ is independent of $\mathcal{F}_{t-1}$, i.e., $\widehat{\varepsilon}_t$ cannot be predicted from information available at time $t-1$. Then, $\Var_{t-1} (\widehat Y_t) = \Var_{t-1}(Y_t) + \Var(\widehat{\varepsilon}_t)$.
Hence, choosing the most accurately estimated proxy $\widehat Y_t$ (i.e., one with smallest possible $\Var(\widehat{\varepsilon}_t)$) gives ECPA tests higher (local) power.
Observing a proxy $\widehat Y_t$ closer to $x_t^\ast$ than $Y_t$---in the sense that $\Var_{t-1}(\widehat Y_t)$ is closer to $\Var_{t-1}(x_t^{\ast})(=0)$ than $\Var_{t-1}(Y_t)$---seems delusive in this application.

While in the previous example, true GDP growth as the target variable is the best proxy, one can sometimes ``get closer'' to the ideal $x_t^{\ast}$.
To illustrate this, consider the classical situation of forecasting the conditional variance $x_t^\ast = \Var_{t-1}(r_t)$ of the univariate log-return $r_t$ on a risky asset with $\E[r_t\mid\mathcal{F}_{t-1}]=0$. In this case, $x_t^\ast=\E[r_t^2\mid\mathcal{F}_{t-1}]$, such that $Y_t = r_t^2$ is the natural target against which to compare forecasts. However, for a standard diffusion process, as in \cite{ABDL03}, one can show that $x_t^{\ast}=\E_{t-1}[IV_t]$ with $IV_t$ the latent integrated variance.
Let $\widehat Y_t = RV_t$ be the realized variance, which estimates $IV_t$ from HF returns.
Then, under \citet[Proposition 3]{ABDL03} and given that $RV_t$ is an unbiased estimator of $IV_t$, it holds that $\E_{t-1} [Y_t] = \E_{t-1} [\widehat Y_t] = x_t^\ast$.
In this setting, it is already well documented that $\widehat Y_t=RV_t$ exhibits much smaller conditional variance than the squared return $Y_t=r_t^2$. Hence, $RV_t$ should be favored for testing ECPA of volatility forecasts as shown by our Theorem~\ref{thm:2} and Proposition~\ref{prop:VarYt}.
Indeed, this is common practice in the literature when comparing volatility forecasts \citep{BPT01,HL05,KJH05,Laurent2013}. However, our Theorem~\ref{thm:2} and Proposition~\ref{prop:VarYt} are the first theoretical results supporting this practice. 
(We mention that \citet{AB98} provide early arguments for the use of HF proxies in the \textit{absolute} evaluation of volatility forecasts via Mincer--Zarnowitz regressions, whereas our focus is on forecast \textit{comparison}.)
Since applications in finance comparing RV to the squared return are already available in the literature, we focus on a macroeconomic application in Section~\ref{Empirical Application}.


Summing up our results of Section \ref{Main Results}, we find that for conditional mean forecasts, ECPA tests based on $\widehat{T}_n$ are asymptotically valid under $H_0$ as long as the proxy $\widehat{Y}_t$ is conditionally unbiased. However, Theorem~\ref{thm:2} and Proposition~\ref{prop:VarYt} show that ECPA tests lose local power when $\Var_{t-1}(\widehat Y_t)>\Var_{t-1}(Y_t)$,  (as is the case in the macroeconomic example) or gain local power when $\Var_{t-1}(\widehat Y_t)<\Var_{t-1}(Y_t)$ (as in the finance application). 
In either case, we urge applied researchers to compare forecasts based on the most accurate proxy available. The next section illustrates the power loss from using imprecise proxies in simulations.

\section{Numerical Experiments}\label{Simulations}

\subsection{Simulations for Robust Loss Functions}\label{Simulations for robust loss functions}

We let $k=1$ and consider the following data-generating process (DGP) and conditional mean forecasts:
\begin{align}
	Y_t 	&= \mu(1-\phi)+\phi Y_{t-1}+\varepsilon_t,&&\varepsilon_t\overset{\text{i.i.d.}}{\sim}N(0,\sigma_{\varepsilon}^2),\label{eq:1}\\
		\widehat{Y}_t &= Y_{t}+\widehat{\varepsilon}_{t},&&\widehat{\varepsilon}_{t}\overset{\text{i.i.d.}}{\sim}N(0,\sigma_{\widehat{\varepsilon}}^2),\label{eq:2}\\
	x_{1t} &= \mu(1-\phi)+\phi Y_{t-1}+\varepsilon_{1,t-1},&&\varepsilon_{1t}\overset{\text{i.i.d.}}{\sim}N(0,\sigma_{1}^2),\label{eq:3}\\
	x_{2t} &= \phi Y_{t-1},&&\label{eq:4}
\end{align}
were $\{\varepsilon_t\}$, $\{\widehat{\varepsilon}_t\}$ and $\{\varepsilon_{1t}\}$ are mutually independent of each other, $|\phi|<1$, $\mu\in\mathbb{R}$, and $\sigma_{1}^2=\mu^2(1-\phi)^2+\xi/\sqrt{n}$ for $\xi\in\mathbb{R}$. Note that the forecasts are measurable with respect to $\mathcal{F}_{t-1}=\sigma(\widehat{Y}_{t-1},\widehat{Y}_{t-2},\ldots;\varepsilon_{1,t-1},\varepsilon_{1,t-2},\ldots;\widehat{\varepsilon}_{t-1},\widehat{\varepsilon}_{t-2},\ldots)$. Of course, it is a convenience (used to render some subsequent computations more tractable) to assume that the latent target $Y_{t-1}$ is known to the forecasters issuing $x_{1t}$ and $x_{2t}$. The main point of these simulations is to investigate the power loss in the comparison from using the proxy $\widehat{Y}_t$. Note that with our choice of $\mathcal{F}_{t-1}$, we have $\E_{t-1}[\widehat{Y}_{t}]=\E_{t-1}[Y_t]$.

We first consider the squared error (SE) as a robust loss function, i.e., $L(Y,x)=(Y-x)^2$, which arises for $\phi(x)=x^2$ and $a(Y)=-Y^2$ in \eqref{eq:L char}. As the SE loss elicits the mean, we have $x_t^\ast=\E_{t-1}[Y_t]=\mu(1-\phi)+\phi Y_{t-1}$. Hence, $x_{1t}$ equals the optimal forecast confounded by some additive noise with variance $\sigma_1^2$. This implies that the larger $\sigma_1^2$, the worse the forecast $x_{1t}$. On the other hand, $x_{2t}$ is a biased forecast, yet with no additive noise. Note that when $\sigma_1^2<\mu^2(1-\phi)^2$ ($\sigma_1^2>\mu^2(1-\phi)^2$), $x_{1t}$ ($x_{2t}$) is the preferred forecast; see \eqref{eq:Sim1}. 

Following \citet{GW06}, we choose the $\mathcal{F}_{t-1}$-predictable test function $h_{t-1}=(1, \widehat{Y}_{t-1})^\prime$. To assess the magnitude of the local alternative, Appendix~\ref{Technical Derivations} shows that 
\begin{equation}\label{eq:Sim1}
\E[h_{t-1}d(Y_t, x_{1t}, x_{2t})] = \begin{pmatrix}\E[d(Y_t, x_{1t}, x_{2t})]\\ \E[\widehat{Y}_{t-1}d(Y_t, x_{1t}, x_{2t})]\end{pmatrix} =\begin{pmatrix}\sigma_1^2 - \mu^2(1-\phi)^2\\ \E[\widehat{Y}_{t-1}]\E[d(Y_t, x_{1t}, x_{2t})]\end{pmatrix} = \begin{pmatrix}\xi/\sqrt{n}\\ \mu\xi/\sqrt{n}\end{pmatrix}.
\end{equation}
Thus, we simulate under the local alternative 
\[
	H_{a,\loc}\ :\ \E[h_{t-1}d(Y_t,x_{1t}, x_{2t})]=\delta/\sqrt{n}, 
\]
where $\delta=(\xi, \mu\xi)^\prime$. For $\xi=0$, our results correspond to size. To compute the theoretical ALP in \eqref{eq:ALP}, we calculate $\Omega$ from D3 in Proposition~\ref{prop:Omega} of Appendix~\ref{Technical Derivations}.

We run our simulations on a grid of $\xi$ values in the interval $[-4,4]$, and compare the empirical rejection frequencies based on $10,000$ simulation replications with the theoretical ALP at a significance level of $5\%$.
We do so for $\mu=1$, $\phi=0.2$, and $\sigma_{\varepsilon}^2=1$. To investigate the impact of noise in the proxy, we consider several values of $\sigma_{\widehat{\varepsilon}}^2$. Specifically, we choose $\sigma_{\widehat{\varepsilon}}^2$ such that $\zeta=\Var(Y_t)/\sigma_{\widehat{\varepsilon}}^2\in\{1/5,\, 1/2,\, 1,\, 2,\, 5,\, \infty\}$ with $\zeta=\infty$ corresponding to the case of no noise, i.e., $\sigma_{\widehat{\varepsilon}}^2=0$.
The ratio $\zeta$ can be interpreted as the \textit{signal to noise ratio} (SNR) of the noisy target; for $\zeta=\infty$, we accurately observe the target, while for $\zeta=1$, the signal $Y_t$ is of the same magnitude as the noise $\widehat{\varepsilon}_t$.

\begin{figure}
	\centering
	\caption{Rejection frequencies for ECPA tests using SE loss and different $\xi$.}		
	\includegraphics[width=\textwidth]{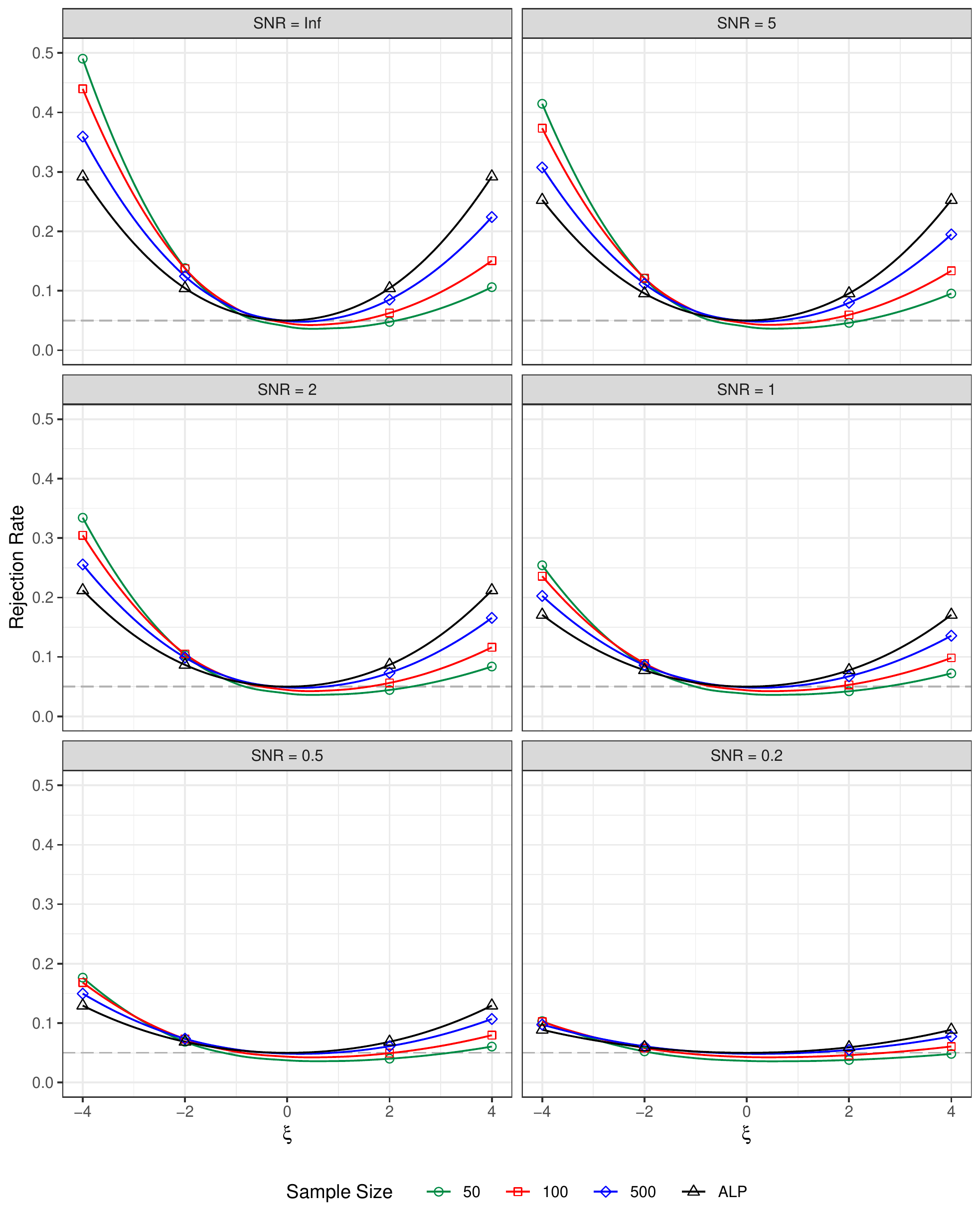}
	\caption*{\footnotesize \textit{Notes:} Empirical rejection frequencies of $\widehat{T}_n$-based ECPA test for sample sizes $n\in\{50,\, 100,\, 500\}$. Results based on DGP and forecasts from \eqref{eq:1}--\eqref{eq:4} with $\mu=1$, $\phi=0.2$, $\sigma_{\varepsilon}^2=1$, $\sigma_1^2=\mu^2(1-\phi)^2+\xi/\sqrt{n}$, and $\sigma_{\widehat{\varepsilon}}^2$ chosen to yield SNR of $\zeta\in\{1/5,\, 1/2,\, 1,\, 2,\, 5,\, \infty\}$. Forecasts are evaluated using SE loss.}
	\label{fig:SimMSE}
\end{figure}

The six panels in Figure~\ref{fig:SimMSE} correspond to the six different values of $\zeta$. In each panel, the empirical rejection frequencies are displayed as a function of $\xi$ for sample sizes $n\in\{50,\ 100,\ 500\}$. The theoretical ALP is displayed as the black reference line. Note that by \eqref{eq:ALP} the ALP curve is symmetric around $\delta=0$ and, hence, also around $\xi=0$. We draw two conclusions from Figure~\ref{fig:SimMSE}.

\begin{enumerate}
	\item As predicted by Theorem~\ref{thm:2}, the empirical rejection frequencies converge to the ALP curve as $n\to\infty$. In particular, no matter how noisy the proxy, size is approximately accurate. As could be expected from a local power result, the rejection frequencies do not increase monotonically with the sample size as is typically the case for fixed alternatives. Here, while for larger $n$ rejections for $\xi>0$ are more frequent in Figure~\ref{fig:SimMSE}, for smaller $n$ rejections for $\xi<0$ occur more often.
	
	\item As also expected from Theorem~\ref{thm:2}, more precise proxies (i.e., with higher $\zeta$) lead to easier discrimination between forecasts in the sense of higher power; see in particular the upper two panels. On the other hand, when $\sigma_{\widehat{\varepsilon}}^2\to\infty$, the signal in $Y_t$ is eventually buried by the noise in $\widehat{\varepsilon}_t$, giving only close to trivial power in the lower two panels.
\end{enumerate}

To summarize, our simulations show that for noisy target variables, ECPA tests are \textit{valid} when using \textit{robust} loss functions in the sense of unaffected rejection rates under the null hypothesis. However, even for robust loss functions, noisy targets negatively influence the test power, which suggests that applied researchers should utilize the most accurate proxy in forecast comparisons.

\subsection{Simulations for Non-Robust Loss Functions}\label{Simulations for non-robust loss functions}

Here, we illustrate the behavior of the ECPA test based on $\widehat{T}_n$ for a non-robust loss function. We again consider the DGP and the two forecasts from \eqref{eq:1}--\eqref{eq:4}.  However, we now use the AE as our loss function $L(\cdot, \cdot)$, which is well-known to elicit the median \citep[Theorem~9]{Gne11}. Since the mean and the median coincide for the symmetric conditional distribution of our employed DGP, the optimal forecast is again $x_t^\ast=\mu(1-\phi)+\phi Y_{t-1}$. We deliberately choose a simulation setting where the mean coincides with the median and, hence, both the SE and AE loss elicit the mean. This allows us to specifically focus on \textit{robustness} of the losses, detached from their elicited functional.

\vspace{0.5cm}

For the AE loss differences we obtain from straightforward calculations using properties of the folded normal distribution that
\begin{align}
	\E[d(Y_t, x_{1t}, x_{2t})] &= (\sigma_{\varepsilon}^2+\sigma_{1}^2) \sqrt{\frac{2}{\pi}} - \sigma_{\varepsilon}^2\sqrt{\frac{2}{\pi}}\exp\left\{-\frac{\mu^2(1-\phi)^2}{2\sigma_{\varepsilon}^2}\right\}\notag\\
	&\hspace{5cm}-\mu(1-\phi)\left[1-2\Phi\left(-\frac{\mu(1-\phi)}{\sigma_{\varepsilon}}\right)\right],\label{eq:d MAE}
\end{align}
where $\Phi(\cdot)$ denotes the standard normal distribution function. Except for $\sigma_1^2$ and $\sigma_{\widehat{\varepsilon}}^2$, we choose the parameters as in the previous subsection, i.e., $\mu=1$, $\phi=0.2$, and $\sigma_{\varepsilon}^2=1$. We vary $\sigma_{\widehat{\varepsilon}}^2$ to yield SNR parameters $\zeta\in\{2,\infty\}$. 
For $\sigma_1^2=0.70...$, we obtain $\E[d(Y_t, x_{1t}, x_{2t})]=0$ by \eqref{eq:d MAE}. Thus, for $\sigma_1^2=0.70...+\xi/\sqrt{n}$ we get that
\begin{align*}
\E[h_{t-1}d(Y_t, x_{1t}, x_{2t})]&=\begin{pmatrix}\E[d(Y_t, x_{1t}, x_{2t})]\\
\E[\widehat{Y}_{t-1}d(Y_t, x_{1t}, x_{2t})]\end{pmatrix}= \begin{pmatrix}\sqrt{2/\pi}\xi/\sqrt{n}\\
\sqrt{2/\pi}\mu\xi/\sqrt{n}\end{pmatrix}=:\delta/\sqrt{n}.
\end{align*}
This again amounts to simulating under $H_{a,\loc}$, with $\xi=0$ corresponding to the null. We vary $\xi$ in the interval $[-4, 4]$.

\begin{figure}[t!]
	\centering
	\caption{Rejection frequencies for ECPA tests using AE loss and different $\xi$.}
		\includegraphics[width=\textwidth]{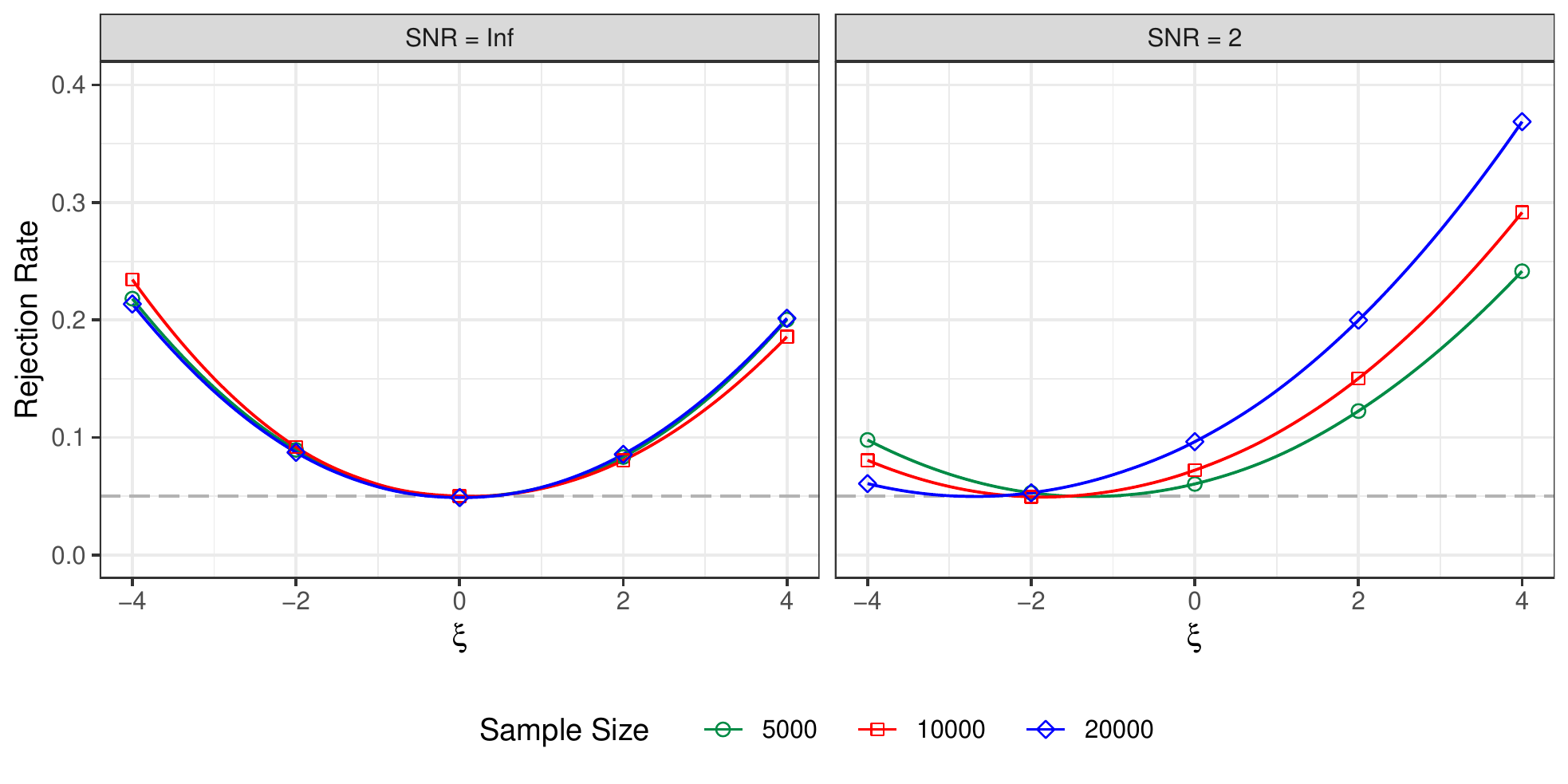}
	\caption*{\footnotesize \textit{Notes:} Empirical rejection frequencies of $\widehat{T}_n$-based ECPA test for sample sizes $n\in\{5000,\, 10000,\, 20000\}$. Results based on DGP and forecasts from \eqref{eq:1}--\eqref{eq:4} with $\mu=1$, $\phi=0.2$, $\sigma_{\varepsilon}^2=1$, $\sigma_1^2=0.70...+\xi/\sqrt{n}$, and $\sigma_{\widehat{\varepsilon}}^2$ chosen to yield SNR of $\zeta\in\{2,\infty\}$. Forecasts are evaluated using AE loss.}
	\label{fig:SimMAE}
\end{figure}

Due to the non-robustness of the AE, we obtain different values for the expectations when $\widehat{Y}_t$ is used. E.g., for $\sigma_1^2=0.70...$ we obtain for $\zeta=2$ that
\begin{equation}\label{eq:(13p)}
	\E[d(\widehat{Y}_t, x_{1t}, x_{2t})]=0.0050... \neq 0=\E[d(Y_t, x_{1t}, x_{2t})].
\end{equation}
Since the difference is rather small, we carry out the simulations for large sample sizes $n\in\{5000,\ 10000,\ 20000\}$ to highlight the implications of \eqref{eq:(13p)} under the null.

If the AE loss were robust, a test based on $\widehat{T}_n$ would have ALP given in \eqref{eq:ALP}. Thus, the empirical local power curve would closely resemble the U-shape of the ALP curves in Figure~\ref{fig:SimMSE}, with a minimum in $\xi=0$ at the nominal level. As reference lines, the left panel of Figure~\ref{fig:SimMAE} shows the test results based on $T_n$, which uses the true value $Y_t$ instead of the proxy $\widehat{Y}_t$, i.e., $\zeta=\infty$. The results based on $T_n$ have the characteristic U-shape, as could be expected from Theorem~\ref{thm:2} applied for $Y_t=\widehat{Y}_t$. However, as the AE loss is non-robust, qualitative differences emerge when using $\widehat{T}_n$, as shown in the right panel of Figure~\ref{fig:SimMAE} ($\zeta=2$). First, the tests are oversized under the null ($\xi=0$), which only gets worse for increasing $n$. Thus, even in this setting where the AE is a judicious choice for evaluating median forecasts, employing the noisy observations $\widehat{Y}_t$ impairs the empirical test size, which contrasts with the SE loss results in Figure~\ref{fig:SimMSE}. Second, and perhaps more seriously, the non-robust loss function leads to decreased power for negative $\xi$, i.e., for $\sigma_1^2<0.7$. We again have the undesirable result that this power loss \textit{increases} in $n$.

\section{Comparing Forecasts for US GDP Growth}
\label{Empirical Application}

As a measure of total real activity, GDP is arguably the most important macroeconomic indicator. However, as pointed out in the Motivation, GDP---and, by extension, GDP growth---can only be measured with error. 
In this application, we focus on continuously-compounded growth rates for US GDP, denoted by $Y_t=\dGDP_t$.
Specifically, we compare SPF and Greenbook forecasts issued at different horizons, yet for the same quarter $t$. 
We evaluate the forecasts using six different GDP growth proxies and we investigate, if---as indicated by Theorem~\ref{thm:2} and Proposition \ref{prop:VarYt}---it is indeed easier to discriminate between two forecasts when a more precise proxy is used.

To describe the growth proxies, recall that there are three ways to compute GDP: the production, income, and expenditures approach; see \citet[Table~1]{LSF08}.
All these methods may be regarded as providing estimates of the true latent value of GDP. 
We use proxies based on the income or expenditures approach, or a combination of the two.

Our first four $\dGDP$ proxies are based on expenditure-side estimates, obtained from \url{https://tinyurl.com/43fk9ms2}.
For this approach, the Federal Reserve Bank of Philadelphia provides data for all available \textit{vintages}, i.e., in each quarter, they report updated $\dGDP$ estimates for all previous quarters.
As information on past GDP accumulates over time, it is reasonable to assume that more recent vintages estimate true growth $\dGDP$ more accurately.
Here, we use the first, second, third, and most recent vintage as our $\widehat{Y}_t$'s, and denote them by $\dGDP_{E1}$, $\dGDP_{E2}$, $\dGDP_{E3}$ and $\dGDP_{E}$, respectively.
The most recent vintage refers to the latest available data as of September 30, 2020. 

Our fifth growth proxy $\dGDP_I$ is based on the most recent vintage of the income approach, and is obtained from \url{https://tinyurl.com/4tjvcb8w}.
While proxies based on the income method feature less prominently in economics, \citet{Nal10} nonetheless finds $\dGDP_I$ to better reflect the growth in real economic activity during the business cycle. 
As \cite{Aea16} note in their conclusion that early vintages of $\dGDP_I$ often provide less accurate information than their counterparts from the expenditure side (mainly due to the delayed availability of accurate tax returns), we only use the most recent vintage $\dGDP_I$.

As our sixth proxy, we use the ``GDP Plus‘‘ approach of \citet{Aea16}, which combines estimates from the income and expenditure side, and which is available at \url{https://tinyurl.com/4tjvcb8w}.
The authors argue that $\dGDP_I$ and $\dGDP_E$ provide complementary information on GDP growth.
The recent popularity of $\dGDP_+$ is documented by the Federal Reserve Bank of Philadelphia reporting it alongside the more classical measures $\dGDP_E$ and $\dGDP_I$.
The methodology behind $\dGDP_+$ has also spurred the development of new GDP proxies \citep{Almuzara2021,Jacobs2020}.

In more detail, \citet{Aea16} propose a dynamic factor model, where $\dGDP_E$ and $\dGDP_I$ both load on the single (latent) factor $\dGDP$:
\begin{equation}\label{eq:GDP}
	\begin{pmatrix}
		\dGDP_{E,t}\\ \dGDP_{I,t}
	\end{pmatrix}=\begin{pmatrix}
		1\\ 1
	\end{pmatrix}\dGDP_{t}+\begin{pmatrix}
		\varepsilon_{E,t}\\ \varepsilon_{I,t}
	\end{pmatrix},
\end{equation}
where $(\varepsilon_{E,t}, \varepsilon_{I,t})^\prime$ are assumed to be i.i.d.~with mean zero, implying in particular conditional unbiasedness of the growth proxies, i.e., $\E_{t-1}[\dGDP_{E,t}]=\E_{t-1}[\dGDP_{I,t}]=\E_{t-1}[\dGDP_{t}]$.
The authors extract GDP growth estimates, denoted $\dGDP_+$, by using the Kalman smoother. 
\citet{Aea16} argue that the Kalman filter extractions $\dGDP_+$ can be regarded as more precise approximations of true growth than either $\dGDP_E$ or $\dGDP_I$. 
As for $\dGDP_I$, we only use the most recent vintage of $\dGDP_+$.

\begin{figure}[tb]
	\begin{center}
		\caption{{GDP Measurements}}
		\begin{subfigure}{\linewidth}
			\centering
			\caption{GDP Vintages based on Expenditure Approach}
			\includegraphics[width=\textwidth]{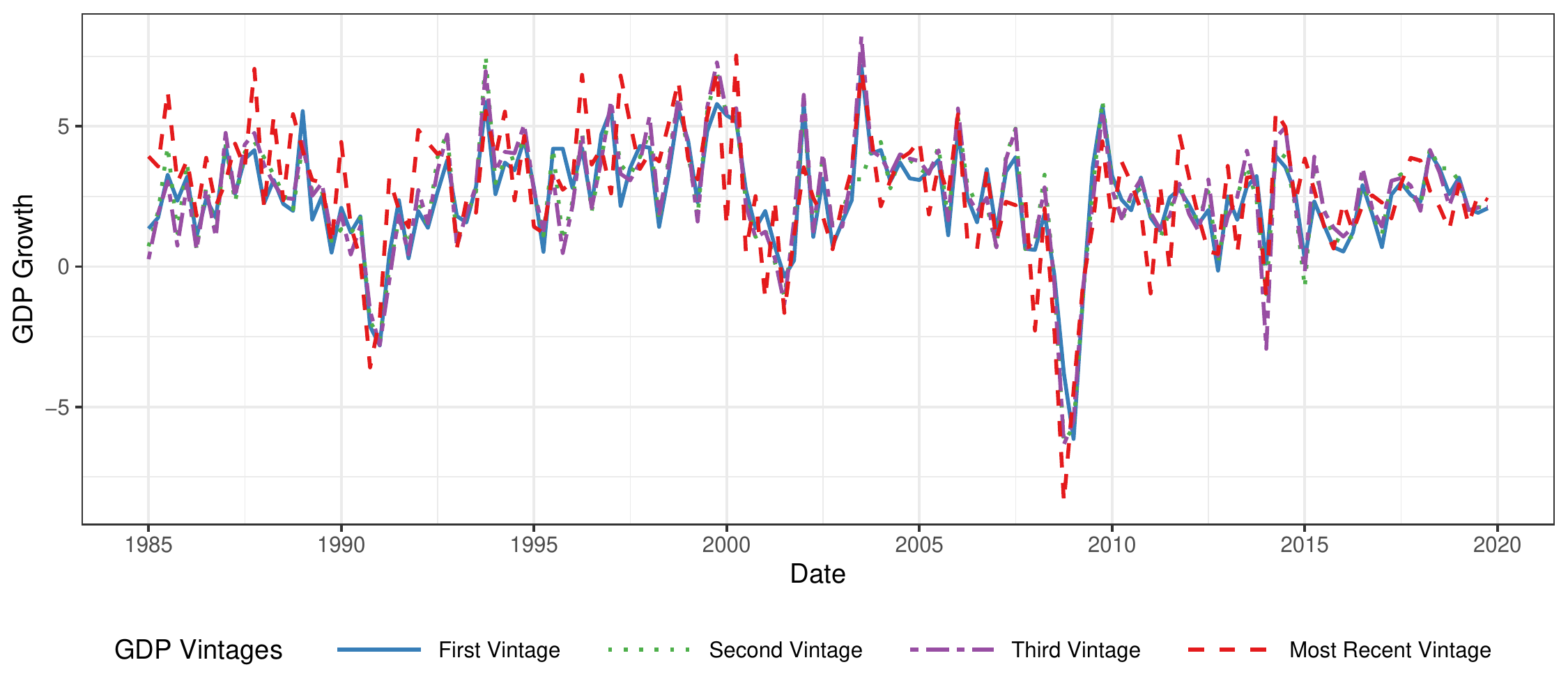}
			\label{fig:GDP_vintages}
		\end{subfigure}
		\begin{subfigure}{\linewidth}
			\centering
			\caption{GDP Measurement Approaches}
			\includegraphics[width=\textwidth]{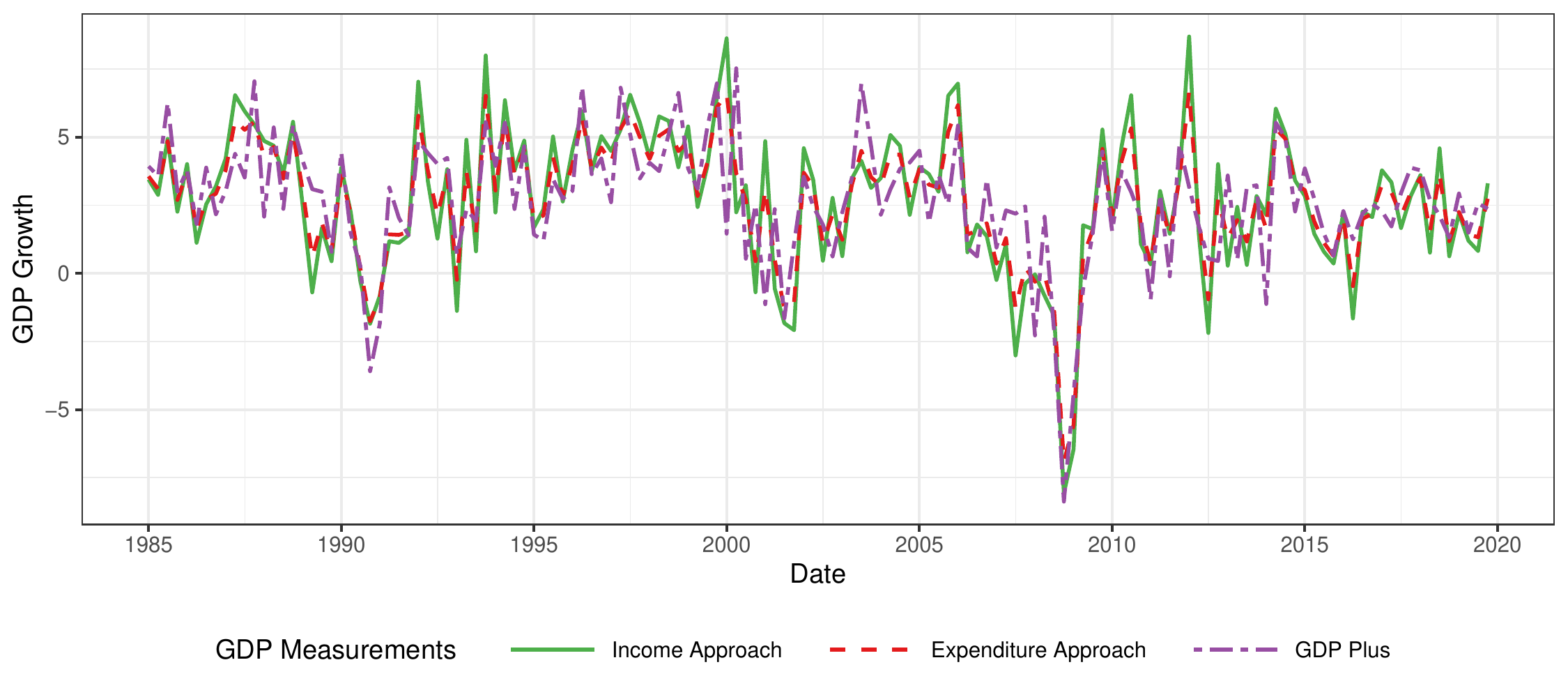}
			\label{fig:GDP_measurements}
		\end{subfigure}
		\label{fig:GDP_plots}
	\end{center}
\end{figure}



We consider the six GDP proxies from 1985Q1 until 2019Q3 to avoid structural breaks in our sample due to the Great Moderation and the Corona crisis. 
%
One may wonder whether the differences in our proxies are large enough during this period to suspect substantially altered results of ECPA tests. To shed light on this, Figure~\ref{fig:GDP_plots} displays the six GDP proxies over time. The upper panel displays the four vintages of the expenditure approach. The lower panel shows $\dGDP_{I}$, $\dGDP_{E}$ and $\dGDP_+$. In both panels, we see a clear joint behavior while the exact measurements differ---sometimes even substantially. For instance, in 2008Q1 in panel~(a) the first three vintages suggest a growing economy, whereas the most recent vintage indicates an almost 2.5\% decline. Hence, the results of ECPA tests may plausibly depend on which proxy is used in the comparison.

Valid ECPA tests using Theorem~\ref{thm:2} rely on \textit{conditionally unbiased} proxies satisfying
\begin{equation}\label{eq:ubi proxy}
	\E_{t-1}[\dGDP_t - \dGDP_{M,t}]\overset{\text{a.s.}}{=}0\qquad\text{for all}\quad M\in\mathcal{M}:=\{E1,\, E2,\, E3,\, E,\, I,\, +\}.
\end{equation}
As true $\dGDP$ is latent, this cannot be tested directly.
Instead, we test the implication of \eqref{eq:ubi proxy} that the proxy differences have zero conditional mean
\begin{equation}\label{eq:proxy impl}
	\E_{t-1}[\dGDP_{M,t}-\dGDP_{N,t}]\overset{\text{a.s.}}{=}0\qquad\text{for all}\quad M\neq N,\ M,N\in\mathcal{M}.
\end{equation}
Of course, it may be the case that all proxies $\dGDP_{M}$ ($M\in\mathcal{M}$), while satisfying \eqref{eq:proxy impl}, are biased in the same direction, thus invalidating \eqref{eq:ubi proxy}. However, as our proxies are based on inherently different approaches to GDP measurement, a common bias seems unlikely. Thus, we view a passed test of \eqref{eq:proxy impl} as a strong indication for \eqref{eq:ubi proxy} to hold as well.

We test \eqref{eq:proxy impl} by using multiple sets of instruments in a standard conditional moment test.
As the sequence $\{\dGDP_{M,t}-\dGDP_{N,t}\}$ is a martingale difference under the null hypothesis in \eqref{eq:proxy impl}, we follow \citet[Comment~5]{GW06} and base our test on the sample variance estimator (instead of on a HAC estimator). 
We use the following five instrument choices: (1) a constant; (2) a constant and the lagged first GDP proxy; (3) a constant and the lagged second GDP proxy; (4) a constant plus the difference of the lagged GDP proxies; and (5) a constant, the first lagged GDP proxy, and  the difference of the lagged GDP proxies.


\begin{table}[tb]
	\caption{$p$-Values of Tests for Conditional Unbiasedness of GDP Proxies}
	\label{tab:GDP_CondUnbiased}
	\centering
	\small
	\begin{tabularx}{0.76\linewidth}{XX @{\hspace{0.5cm}} l rrrrr}
		\toprule
		& & & \multicolumn{5}{c}{Instruments}  \\
		\cmidrule(lr){3-8}
		$GDP_M$	& $GDP_N$ & & Inst.~1 &  Inst.~2 &  Inst.~3 &  Inst.~4 &  Inst.~5 \\
		\midrule
		$\dGDP_{E1}$ & $\dGDP_{E2}$ &   & 0.363 & 0.516 & 0.527 & 0.825 & 0.706\\
$\dGDP_{E1}$ & $\dGDP_{E3}$ &   & 0.227 & 0.379 & 0.269 & 1.000 & 0.590\\
$\dGDP_{E1}$ & $\dGDP_{E}$  &   & 0.106 & 0.089 & 0.276 & 0.551 & 0.635\\
$\dGDP_{E1}$ & $\dGDP_{I}$  &   & 0.171 & 0.061 & 0.106 & 0.978 & 0.524\\
$\dGDP_{E1}$ & $\dGDP_{+}$  &   & 0.097 & 0.045 & 0.031 & 0.982 & 0.350\\
$\dGDP_{E2}$ & $\dGDP_{E3}$ &   & 0.649 & 0.892 & 0.878 & 0.987 & 0.962\\
$\dGDP_{E2}$ & $\dGDP_{E}$  &   & 0.244 & 0.289 & 0.532 & 0.627 & 0.797\\
$\dGDP_{E2}$ & $\dGDP_{I}$  &   & 0.306 & 0.205 & 0.226 & 0.998 & 0.622\\
$\dGDP_{E2}$ & $\dGDP_{+}$  &   & 0.237 & 0.208 & 0.105 & 0.937 & 0.486\\
$\dGDP_{E3}$ & $\dGDP_{E}$  &   & 0.308 & 0.421 & 0.613 & 0.825 & 0.918\\
$\dGDP_{E3}$ & $\dGDP_{I}$  &   & 0.377 & 0.350 & 0.321 & 0.993 & 0.778\\
$\dGDP_{E3}$ & $\dGDP_{+}$  &   & 0.321 & 0.389 & 0.199 & 0.843 & 0.698\\
$\dGDP_{E}$  & $\dGDP_{I}$  &   & 0.862 & 0.407 & 0.974 & 0.246 & 0.986\\
$\dGDP_{E}$  & $\dGDP_{+}$  &   & 0.935 & 0.574 & 0.962 & 0.282 & 0.988\\
$\dGDP_{I}$  & $\dGDP_{+}$  &   & 0.736 & 0.612 & 0.335 & 0.646 & 0.808\\
	
		\bottomrule			
		\addlinespace
		\multicolumn{8}{p{.73\linewidth}}{\footnotesize \textit{Notes:} This table presents $p$-values for the tests of conditional unbiasedness of all 15 combinations of the employed GDP proxies. The five columns refer to the instrument choices given in the main text.
	}
	\end{tabularx}
\end{table}

Table~\ref{tab:GDP_CondUnbiased} reports $p$-values of these conditional moment restriction tests for all 15 distinct pairwise combinations of the six GDP proxies.
We find that at the $5\%$ ($10\%$) significance level, the null hypothesis can only be rejected in 2 (5) out of the 75 cases, which is below the nominal level for both choices.

We now describe the forecasts that we compare using our six proxies $\widehat{Y}_t$. The first set of growth forecasts is taken from the SPF. Since 1968, the SPF publishes quarterly forecasts---prepared by private sector economists---of macroeconomic variables in the US. The SPF is widely used for forecast comparisons in the academic literature \citep[e.g.,][]{Car03,Cam07,EMW09}. Specifically, we consider the (cross-respondent) mean GDP growth predictions of the SPF, available at \url{https://tinyurl.com/y9p8oylx}.

Second, we employ the so-called “Greenbook” forecasts of GDP growth, available at \url{https://tinyurl.com/y7b6pm2f}.
By using substantial resources, the staff of the Board of Governors of the Federal Reserve prepares these forecasts for each meeting of the Federal Open Market Committee \citep{RomerRomer2000}.
We use the forecast closest to the middle of each of the respective quarters.
Notice that the Greenbook forecasts are only available to the public after a five-year lag.

Our goal is to show that ECPA tests can more easily discriminate between two forecasts $x_{1t}$ and $x_{2t}$ if a more precise proxy is used. To this end, we require one forecast to be clearly superior to the other. For this, we employ forecasts for the same quarter $t$ (from \textit{either} the SPF \textit{or} Greenbook), however with varying horizons, i.e., forecasts for $\dGDP_t$ issued at different time points $t-\tau$. \citet{HE14} formally show that forecasts based on larger information sets are superior to those based on smaller information sets. 
The forecast with the shorter horizon ($x_{1t}$) naturally nests the information set of the longer horizon forecast ($x_{2t}$) that was issued a longer time ago, implying superiority of the former.
As we consider multi-step ahead forecasts, we employ a HAC covariance estimator \citep{Newey/West:87a, Andrews:91}.

\begin{table}[t!]
	\caption{Loss Differences and $p$-Values of ECPA Tests for Multiple GDP Measurements}
	\label{tab:GDP_ECPA_Results}
	\centering
	\small
	\begin{tabularx}{1\linewidth}{X @{\hspace{1cm}} lrrr @{\hspace{0.6cm}} lrrr @{\hspace{0.6cm}} lrrr}
		\toprule
		& \multicolumn{4}{c}{1Q vs.\ 2Q Ahead} & \multicolumn{4}{c}{1Q vs.\ 4Q Ahead}  & \multicolumn{4}{c}{2Q vs.\ 4Q Ahead} \\
		\cmidrule(lr){2-5}	\cmidrule(lr){6-9}	 \cmidrule(lr){10-13}	
		&& Loss & \multicolumn{2}{c}{$p$-value} && Loss & \multicolumn{2}{c}{$p$-value} && Loss & \multicolumn{2}{c}{$p$-value} \\
		\cmidrule(lr){4-5}	\cmidrule(lr){8-9}	 \cmidrule(lr){12-13}	
		&& Diff. & Inst.~$1$ & Inst.~$2$ && Diff. & Inst.~$1$ & Inst.~$2$ && Diff. & Inst.~$1$ & Inst.~$2$ \\
		\midrule
		\\
		& & \multicolumn{11}{l}{Panel A: Greenbook Forecasts and GDP Vintages} \\
		\cmidrule(lr){2-13}
		$\dGDP_{E1}$ &   & -0.565 & 0.115 & 0.265 &   & -0.874 & 0.075 & 0.208 &   & -0.278 & 0.194 & 0.374\\
$\dGDP_{E2}$ &   & -0.611 & 0.081 & 0.209 &   & -0.835 & 0.106 & 0.301 &   & -0.179 & 0.404 & 0.462\\
$\dGDP_{E3}$ &   & -0.570 & 0.096 & 0.255 &   & -0.823 & 0.107 & 0.310 &   & -0.218 & 0.329 & 0.406\\
$\dGDP_{E}$  &   & -0.554 & 0.091 & 0.206 &   & -0.911 & 0.070 & 0.151 &   & -0.412 & 0.087 & 0.196\\

		\midrule
		\\
		& & \multicolumn{11}{l}{Panel B: SPF Forecasts and GDP Vintages} \\
		\cmidrule(lr){2-13}
		$\dGDP_{E1}$ &   & -0.551 & 0.127 & 0.344 &   & -1.146 & 0.111 & 0.158 &   & -0.595 & 0.116 & 0.221\\
$\dGDP_{E2}$ &   & -0.508 & 0.157 & 0.406 &   & -1.234 & 0.101 & 0.135 &   & -0.725 & 0.075 & 0.078\\
$\dGDP_{E3}$ &   & -0.464 & 0.197 & 0.432 &   & -1.156 & 0.126 & 0.267 &   & -0.692 & 0.098 & 0.145\\
$\dGDP_{E}$  &   & -0.509 & 0.118 & 0.251 &   & -1.286 & 0.094 & 0.129 &   & -0.777 & 0.091 & 0.151\\
	
		\midrule
		\\
		& & \multicolumn{11}{l}{Panel C: Greenbook Forecasts and GDP Measurements} \\
		\cmidrule(lr){2-13}
		$\dGDP_I$ &   & -0.665 & 0.168 & 0.414 &   & -0.770 & 0.178 & 0.441 &   & -0.178 & 0.427 & 0.302\\
$\dGDP_E$ &   & -0.554 & 0.091 & 0.206 &   & -0.911 & 0.070 & 0.151 &   & -0.412 & 0.087 & 0.196\\
$\dGDP_{+}$ &   & -0.555 & 0.152 & 0.286 &   & -0.747 & 0.160 & 0.391 &   & -0.250 & 0.189 & 0.202\\
	
		\midrule
		\\
		& & \multicolumn{11}{l}{Panel D: SPF Forecasts and GDP Measurements} \\
		\cmidrule(lr){2-13}
		$\dGDP_I$ &   & -0.557 & 0.117 & 0.245 &   & -1.304 & 0.092 & 0.096 &   & -0.747 & 0.093 & 0.078\\
$\dGDP_E$ &   & -0.509 & 0.118 & 0.251 &   & -1.286 & 0.094 & 0.129 &   & -0.777 & 0.091 & 0.151\\
$\dGDP_{+}$ &   & -0.531 & 0.093 & 0.182 &   & -1.258 & 0.086 & 0.063 &   & -0.727 & 0.086 & 0.129\\
			
		\bottomrule			
		\addlinespace
		\multicolumn{13}{p{.98\linewidth}}{\footnotesize \textit{Notes:} This table reports the loss differences (column``Loss Diff.'') together with the $p$-values of the tests of ECPA based on the two instrument choices given in the text for the different comparisons given in Panels A--D.
		The first column panel reports results for the comparison of one- against two-quarter ahead forecasts, whereas the second and third column panels compare one- against four-, and two- against four-quarter ahead forecasts.
		}
	\end{tabularx}
\end{table}

Table~\ref{tab:GDP_ECPA_Results} shows $p$-values of ECPA tests (based on $\widehat{T}_n$) for the three combinations of one-, two-, and four-quarter ahead GDP growth forecasts. It does so individually for the SPF and the Greenbook forecasts.
As test functions $h_{t-1}$, we use a constant only (Inst.~1), and a constant jointly with the loss difference of the forecast, lagged by the horizon of the shorter of the two forecast horizons (Inst.~2).
Panels~A and~B of Table~\ref{tab:GDP_ECPA_Results} report results for the four different vintages, while Panels~C and~D are for $\dGDP_{E}$, $\dGDP_{I}$ and $\dGDP_{+}$.

We find that all loss differences in Table~\ref{tab:GDP_ECPA_Results} are negative, implying---as expected---better predictive ability of the shorter horizon forecasts.
For the upper two panels based on the different GDP vintages, we find substantially lower $p$-values for more recent vintages, which we explain by the theoretically higher test power.
E.g., the most recent vintage exhibits a smaller $p$-value than the first vintage in all 12 instances; and it attains overall the smallest $p$-value in nine of the 12 cases.

The results are qualitatively similar---but less pronounced---for the lower two panels of Table~\ref{tab:GDP_ECPA_Results}, which analyze $\dGDP_{E}$, $\dGDP_{I}$ and $\dGDP_{+}$.
Here, the $\dGDP_+$ approach, which is claimed to be superior in the literature, exhibits the lowest $p$-value only in five instances, but it seems to be generally lower, especially compared to the tests based on $\dGDP_{I}$.

\section{Conclusion}\label{Conclusion}

In this paper, we answer the following question: what can we learn from  forecast comparisons if the target variable cannot be observed precisely, and only mismeasured proxy variables are available?
We show that the classical forecast evaluation tools of loss functions and inference thereon can be used without modification if (a) a conditionally unbiased proxy is available and (b) the target functional is the conditional \textit{mean}.
In contrast, for other target functionals such as e.g., the conditional median, the use of approximated proxy variables generally distorts the loss differences and hence, inference in tests for ECPA.

This leads to the perhaps surprising conclusion that when evaluating forecasts in the presence of measurement error, the mean is ``more robust'' than the median, whereas the converse is well-known in classical estimation theory.
Hence, using the mean as target functional is particularly attractive in forecasting settings that are prone to measurement error, such as in our empirical application on GDP growth.
Further applications are widespread and include, among others, macroeconomic variables as inflation rates, volatility forecasts in finance, meteorological quantities as precipitation or wind speeds, and case or death counts in infectious disease forecasting.

We further demonstrate that even though standard inference in classical ECPA tests for mean forecasts is valid under measurement error, the test's (local) power decreases with the magnitude of the error.
This gives theoretical content to the empirical observation that ``[a]lthough consistency of the ordering is ensured by an appropriate choice of the loss function independently of the quality of the proxy, a high precision proxy allows to efficiently discriminate between models'' \citep[p.~7]{Laurent2013}. 
We also confirm this in Monte Carlo experiments and provide an empirical illustration on GDP growth.
We emphasize that this increase in power afforded by more precise proxies is particularly important in economics, where sample sizes are often limited---due to low-frequency data collection, structural breaks, etc.

\singlespacing
\putbib[thebib]
\end{bibunit}

\newpage
\setcounter{page}{1}

\doublespacing

\appendix
\appendixpage

\begin{bibunit}

\section{(Non-)Robust Evaluation of Quantile Forecasts}
\label{sec:NonRobustQuantiles}

As already noted after Theorem \ref{thm:1}, it remains an open question if forecasts for other target functionals than the mean can be evaluated robustly in the presence of measurement error when the conditional unbiasedness condition is replaced by some other ``resemblance condition''.
In this section, we focus on the most prominent class of target functionals besides the mean, that is, on conditional quantiles at levels $\alpha \in (0,1)$.

For technical simplicity (in the proof of Proposition \ref{prop:QuantileRobustness} below), we restrict attention to the univariate case ($k=1$), and to target variables with continuous conditional distribution functions and bounded support.
Formally, let $Y_t, \widehat Y_t:\Omega\rightarrow\mathsf{O}\subset\mathbb{R}$, where $\mathsf{O}$ is a bounded interval, and where for all $\omega \in \Omega$, the distribution functions $F_t(\omega, \cdot)$ and $\widehat F_t(\omega, \cdot)$ are continuous.
As customary for conditional quantile forecasts, $x_{t}:\Omega\rightarrow\mathsf{A}$, we set $\mathsf{A} = \mathsf{O}$.

For any strictly increasing function $g: \mathsf{O} \to \mathbb{R}$, we denote by $\mathcal{P}_g$ the class of compactly supported continuous distribution functions, $F$, such that $\E_{Y \sim F}[g(Y)]$ exists and is finite.
Then, the loss functions
\begin{align}
	\label{eqn:QuantLoss}
	L_{g,\alpha}(Y_t, x_t) =  \big( \mathds{1}(x_t \ge Y_t) - \alpha \big) \cdot \big( g(x_t) - g(Y_t) \big), 
\end{align}
are strictly $\mathcal{P}_g$-consistent for the $\alpha$-quantile.
Subject to further technical conditions, \textit{all} strictly consistent loss functions are of the form \eqref{eqn:QuantLoss}; see \citet[Theorem 9]{Gne11} for details.


For two competing forecasts $x_{1t}, x_{2t}: \Omega\rightarrow\mathsf{A}$, let $d_{g,\alpha}({Y}_t, x_{1t}, x_{2t})  = L_{g,\alpha}(Y_t, x_{1t})  - L_{g,\alpha}(Y_t, x_{2t})$.
Straight-forward calculations yield that the expected difference of loss differences (DLD) is
\begin{align}
\begin{aligned}
\label{eqn:QuantLossDiffDiff}
&\E_{t-1} \left[ d_{g,\alpha}({Y}_t, x_{1t}, x_{2t}) -  d_{g,\alpha}(\widehat{Y}_t, x_{1t}, x_{2t}) \right] \\
= \;&\E_{t-1} \big[   g(Y_t)  \big( \mathds{1}(Y_t \le x_{2t})  -  \mathds{1}(Y_t \le x_{1t}) \big)   \big]  
-  \E_{t-1} \big[ g(\widehat Y_t)  \big( \mathds{1}(\widehat Y_t \le x_{2t})  -  \mathds{1}(\widehat Y_t \le x_{1t}) \big)   \big] \\
&\qquad+ g(x_{1t}) \big[ F_t(x_{1t})  -  \widehat F_t(x_{1t}) \big]  - g(x_{2t}) \big[ F_t(x_{2t})  - \widehat F_t(x_{2t}) \big].
\end{aligned}
\end{align}

An exact robustness property would mean that the expected DLD must be zero \textit{for all} forecasts $x_{1t}$ and $x_{2t}$.
E.g., for any Bregman loss function in \eqref{eq:L char}, the expected DLD is linear in $\E_{t-1} [Y_t - \widehat Y_t]$, readily implying exact robustness for conditionally unbiased proxies.
In contrast, the form in \eqref{eqn:QuantLossDiffDiff} is much more complicated. We have the following result, whose proof is in Appendix~\ref{Proofs}:

\begin{prop}
	\label{prop:QuantileRobustness}
	Let $g: \mathsf{O} \to \mathbb{R}$ be a strictly increasing function and let $F_t(\omega, \cdot), \widehat F_t(\omega, \cdot) \in \mathcal{P}_g$ for all $\omega \in \Omega$.
	If
	\begin{align*}
	\E_{t-1} \left[ d_{g,\alpha}({Y}_t, x_{1t}, x_{2t}) -  d_{g,\alpha}(\widehat{Y}_t, x_{1t}, x_{2t}) \right] = 0 \; \text{ a.s.\ for all $\mathcal{F}_{t-1}$-measurable $x_{1t}$ and $x_{2t}$},
	\end{align*}
	then 
	$Y_t = \widehat Y_t$ in distribution.
\end{prop}

Proposition \ref{prop:QuantileRobustness} shows that for quantile forecasts and their corresponding class of strictly consistent loss function in \eqref{eqn:QuantLoss}, an exhaustive exact robustness property---which holds for all possible forecasts $x_{1t}$ and $x_{2t}$---can only hold under the degenerate condition that the entire distributions of $Y_t$ and $\widehat Y_t$ coincide.

However, the specification of the expected DLD in \eqref{eqn:QuantLossDiffDiff} gives rise to a weaker form of robustness:
Equivalence of the conditional distributions $F_t(\cdot)$ and $\widehat F_t(\cdot)$ on the interval between $x_{1t}$ and $x_{2t}$, i.e., on the relevant region between the competing forecasts, is sufficient to guarantee that the expected DLD equals zero.
E.g., for forecasts for the conditional median, measurement error sufficiently far in the tails of the distributions leaves the expected DLD unchanged.
Analogously, for forecasts of quantiles in the left tail of the distribution, as in value-at-risk forecasting, measurement error in the right tail does not affect the loss differences.
However, this weaker form of robustness entails that particularly poor forecasts require stronger relations between $Y_t$ and $\widehat Y_t$, i.e., equivalence of the conditional distributions on larger intervals.

\begin{rem}
	Besides point predictions, interest is often on predictive $\mathcal{F}_{t-1}$-measurable cumulative distribution functions (CDF) of $Y_t$.
	We denote the competing CDF-valued forecasts by $G_{1t}$ and $G_{2t}$.
	A popular (strictly) proper scoring rule is the threshold-weighted CRPS (twCRPS), based on some non-negative weight function $u(z)$ \citep{GR11},
	\begin{align*}
	L(Y_t, G_t) = \int_\mathbb{R} u(z) \left[ G_t(z) - \mathds{1}(Y_t \le z) \right]^2 \mathrm{d}z.
	\end{align*}
	For this scoring rule, the expected DLDs (based on forecasts $G_{1t}$ and $G_{2t}$) is given by
	\begin{align*}
		2 \int_\mathbb{R} u(z) \Big[ G_{2t}(z) - G_{1t}(z) \Big]  \cdot \Big[ F_t(z) - \widehat F_t(z) \Big] \mathrm{d}z,
	\end{align*}
	which is strictly positive almost surely for $G_{1t} = \widehat F_t$ and $G_{2t} = F_t$, unless the conditional distributions of $Y_t$ and $\widehat Y_t$ coincide almost surely (or unless $u(z) = 0$ on the region where the conditional distributions differ, which however violates the score's strict propriety).
	
	This implies that even for the most natural forecasts, i.e., ideal probabilistic forecasts for $F_t$ and $\widehat F_t$, the expected DLD is non-zero.
	Hence, any form of exact robustness seems out of reach for the twCRPS, and more general for probabilistic forecasts, which might explain the absence of strong theoretical results for the evaluation of probability forecasts under measurement error; also see \cite{Fer17}, \cite{Kleen2019} and \citet[Example 6]{BrehmerGneiting2020}.
\end{rem}


\section{Proofs}\label{Proofs}

\renewcommand{\theequation}{B.\arabic{equation}}	
\setcounter{equation}{0}

This appendix collects the proofs of all theoretical results in the main paper and the proof of Proposition~\ref{prop:QuantileRobustness} from Appendix~\ref{sec:NonRobustQuantiles}.

\begin{proof}[{\textbf{Proof of Theorem~\ref{thm:1}:}}]
The theorem is established upon proving $(a)\Longrightarrow(b)\Longrightarrow(c)\Longrightarrow(d)\Longrightarrow(a)$.

We first show $(a)\Longrightarrow(b)$. Write
\begin{align*}
	\E_{t-1}[d(Y_t, x_{1t}, x_{2t})] &= \E_{t-1}[L(Y_t, x_{1t})-L(Y_t, x_{2t})]\\
		&= \E_{t-1}[\phi(x_{1t})-\phi(x_{2t}) + \langle\D\phi(x_{1t}), Y_t-x_{1t}\rangle - \langle\D\phi(x_{2t}), Y_t-x_{2t}\rangle].
\end{align*}
Hence,
\begin{align*}
	\E_{t-1}[d(Y_t, x_{1t}, x_{2t})]-\E_{t-1}[d(\widehat{Y}_t, x_{1t}, x_{2t})] &= \E_{t-1}[\langle\D\phi(x_{1t}), Y_t-\widehat{Y}_t\rangle-\langle\D\phi(x_{2t}), Y_t-\widehat{Y}_t\rangle]\\
	&= \langle\D\phi(x_{1t}),\E_{t-1}[Y_t-\widehat{Y}_t]\rangle-\langle\D\phi(x_{2t}), \E_{t-1}[Y_t-\widehat{Y}_t]\rangle\overset{\text{a.s.}}{=}0,
\end{align*}
by conditional unbiasedness of $\widehat{Y}_t$. This implies $(b)$.

The implication $(b)\Longrightarrow(c)$ is immediate, so we prove $(c)\Longrightarrow(d)$ next. Consider some $Y_t$ and some conditionally unbiased $\widehat{Y}_t$ with $F_t(\omega,\cdot)\in\mathcal{P}$ and $\widehat{F}_t(\omega,\cdot)\in\mathcal{P}$ for all $\omega\in\Omega$. Let $\widetilde{x}_t:\Omega\rightarrow\mathsf{A}$ be a $\mathcal{F}_{t-1}$-measurable random variable that is different from $\omega\mapsto x_t^\ast(\omega)=T(F_{t}(\omega,\cdot))$ with positive probability.
By (a straightforward multivariate extension of) Theorem~1 of \citet{HE14}, $x_t^\ast$ uniquely minimizes $\E_{t-1}[L(Y_t,\cdot)]$ in the probabilistic sense, that is,
\begin{equation}
	\label{eq:prxy le}
	\E_{t-1}[L(Y_t,x_t^\ast)] < \E_{t-1}[L(Y_t,\widetilde{x}_t)]
\end{equation}
with positive probability, and $\E_{t-1}[L(Y_t,x_t^\ast)] \overset{\text{a.s.}}{\leq} \E_{t-1}[L(Y_t,\widetilde{x}_t)]$. The latter inequality and ordering robustness imply
\begin{equation}\label{eq:prxy leq}
	\E_{t-1}[L(\widehat{Y}_t,x_t^\ast)]\overset{\text{a.s.}}{\leq}\E_{t-1}[L(\widehat{Y}_t,\widetilde{x}_t)].
\end{equation}
Proceed by contradiction and suppose that equality in \eqref{eq:prxy leq} holds almost surely. Then,
\[
	\E_{t-1}[L(\widehat{Y}_t,x_t^\ast)]\overset{\text{a.s.}}{\leq}\E_{t-1}[L(\widehat{Y}_t,\widetilde{x}_t)]\qquad\text{and}\qquad\E_{t-1}[L(\widehat{Y}_t,x_t^\ast)]\overset{\text{a.s.}}{\geq}\E_{t-1}[L(\widehat{Y}_t,\widetilde{x}_t)].
\]
By ordering robustness, this is equivalent to 
\[
	\E_{t-1}[L(Y_t,x_t^\ast)]\overset{\text{a.s.}}{\leq}\E_{t-1}[L(Y_t,\widetilde{x}_t)]\qquad\text{and}\qquad\E_{t-1}[L(Y_t,x_t^\ast)]\overset{\text{a.s.}}{\geq}\E_{t-1}[L(Y_t,\widetilde{x}_t)],
\]
which in turn is equivalent to $\E_{t-1}[L(Y_t,x_t^\ast)]\overset{\text{a.s.}}{=}\E_{t-1}[L(Y_t,\widetilde{x}_t)]$. 
However, this contradicts \eqref{eq:prxy le} and thus, the inequality in \eqref{eq:prxy leq} must hold in the strict sense with positive probability.
This implies that $x_t^\ast$ also uniquely minimizes $\E_{t-1}[L(\widehat{Y}_t,\cdot)]$ in the probabilistic sense.

The remainder of the proof of this step is adapted from arguments kindly given to us by Tobias Fissler. Collecting our results so far, we have for $x_t^\ast(\omega)=T(F_{t}(\omega,\cdot))$ that
\begin{equation}\label{eq:star}
	\E_{t-1} \big[ L(Y_t, x_t^\ast) \big] \overset{\text{a.s.}}{\le} \E_{t-1} \big[ L(Y_t, \tilde x_t) \big]  \qquad \text{for all}\ \mathcal{F}_{t-1}\text{-measurable}\ \tilde x_t:\Omega\rightarrow\mathsf{A}.
\end{equation}
Furthermore, for \textit{any} conditionally unbiased proxy $\widehat{Y}_t$ with $\E_{t-1}[\widehat{Y}_t]=\E_{t-1}[Y_t]$ we have that
\begin{equation}\label{eq:star1}
	\E_{t-1} \big[ L(\widehat{Y}_t, x_t^\ast) \big] \overset{\text{a.s.}}{\le} \E_{t-1} \big[ L(\widehat{Y}_t, \tilde x_t) \big]  \qquad \text{for all}\ \mathcal{F}_{t-1}\text{-measurable}\ \tilde x_t:\Omega\rightarrow\mathsf{A},
\end{equation}
and almost sure equality in \eqref{eq:star1} again implies $\tilde x_t \overset{\text{a.s.}}{=} x_t^\ast$. By Theorem~1 of \citet{HE14}, \eqref{eq:star1} implies $x_t^\ast(\omega)=T(\widehat{F}_t(\omega,\cdot))$ for $\p$-almost every (a.e.) $\omega\in\Omega$. Hence, $x_t^\ast(\omega)=T(F_t(\omega,\cdot))=T(\widehat{F}_t(\omega,\cdot))$ for all conditionally unbiased proxies. One functional satisfying this is the mean functional
\begin{equation}\label{eq:meanfunctional}
	T:\mathcal{P}\rightarrow\mathsf{A},\quad F\mapsto T(F)=\int y\D F(y),
\end{equation}
because then $x_t^\ast=\E_{t-1}[Y_t]=\E_{t-1}[\widehat{Y}_t]$ for all conditionally unbiased proxies. 
By uniqueness of the minimizer (in the probabilistic sense), no other functional which deviates from \eqref{eq:meanfunctional} satisfies \eqref{eq:star} and \eqref{eq:star1}. Thus, $(d)$ follows.


Finally, we show that $(d)\Longrightarrow(a)$. From $(d)$ we know that $L(\cdot,\cdot)$ elicits the mean, given in \eqref{eq:meanfunctional}.
Theorem~11 of \citet{FK15} then implies that $L(\cdot,\cdot)$ is of the form given in \eqref{eq:L char} with convex $\phi:\mathsf{A}\rightarrow\mathbb{R}$.

It remains to show that $\phi:\mathsf{A}\rightarrow\mathbb{R}$ is necessarily strictly convex. If $\mathsf{A}$ is a singleton, there is nothing to show. So we assume that $\mathsf{A}$ is not degenerate. We proceed by contradiction and assume that $\phi(\cdot)$ is merely convex, but not strictly convex. Then, by Propositions~6.1.1 and~6.1.3 in \citet{HL01}, there exist $x_1,x_2\in\mathsf{A}$ with $x_1\neq x_2$, such that
\[
	\phi(x_1)-\phi(x_2)+\langle\D\phi(x_1), x_2-x_1\rangle=0.
\]
Now, by surjectivity of $T$, there exists $F\in\mathcal{P}$, such that $T(F)=x_2$. Then, for $Y\sim F$,
\[
	\E_{Y \sim F}[L(Y,x_1)]-\E_{Y \sim F}[L(Y,x_2)] = \phi(x_1)-\phi(x_2)+\langle\D\phi(x_1), x_2-x_1\rangle=0,
\]
which contradicts the strict consistency of $L(\cdot,\cdot)$ for $T$. Hence, $\phi(\cdot)$ must be strictly convex, concluding the proof.
\end{proof}

\begin{proof}[{\textbf{Proof of Theorem~\ref{thm:2}:}}]
The proof follows along similar lines as that of Theorems~1 and~2 in \citet{GW06}. The main difference is that we need to consider triangular arrays. We first prove
\begin{equation}\label{eq:Conv}
 \Omega^{-1/2}\frac{1}{\sqrt{n}}\sum_{t=1}^{n}\widehat{Z}_{n,t}\overset{d}{\longrightarrow}N(\Omega^{-1/2}\delta,I_{q\times q}),\quad\text{as }n\to\infty,
\end{equation}
where $I_{q\times q}$ denotes the $(q\times q)$-identity matrix. 

Before turning to the proof, observe that by D4 there exists some function $f_{n,t}$, such that $\widehat{Z}_{n,t}=f_{n,t}(h_{t-1},\widehat{W}_{t},\ldots,\widehat{W}_{t-m})$. Suppose that $\alpha$-mixing holds. Then, the $\alpha$-mixing coefficients of $\widehat{Z}_{n,t}$ are defined by
\begin{equation}\label{eq:alpha}
	\sup_{n}\alpha_n(h)=\sup_n \sup_{t}\sup_{A\in\mathcal{A}_{n,t},\ B\in\mathcal{B}_{n,t+h}} \big|\p\{A\cap B\}- \p\{A\}\p\{B\}\big|,
\end{equation}
where $\mathcal{A}_{n,t}=\sigma(\ldots,\widehat{Z}_{n,t-1},\widehat{Z}_{n,t})$ and $\mathcal{B}_{n,t+h}=\sigma(\widehat{Z}_{n,t+h},\widehat{Z}_{n,t+h+1},\ldots)$. 
The inclusions 
\begin{align*}
	\mathcal{A}_{n,t}&\subseteq\sigma\big(\ldots,(\widehat{W}_{t-1}^\prime, h_{t-1}^\prime)^\prime, (\widehat{W}_t^\prime, h_t^\prime)^\prime\big),\\
	\mathcal{B}_{n,t+h}&\subseteq\sigma\big((\widehat{W}_{t+h-m}^\prime, h_{t+h-m}^\prime)^\prime, (\widehat{W}_{t+h-m+1}^\prime, h_{t+h-m+1}^\prime)^\prime,\ldots\big)
\end{align*}
hold since $\widehat{Z}_{n,t}=f_{n,t}(h_{t-1},\widehat{W}_{t},\ldots,\widehat{W}_{t-m})$.
Therefore, by \eqref{eq:alpha}, $\sup_{n}\alpha_n(h)\leq\alpha(h-m)$, where $\alpha(\cdot)$ denote the $\alpha$-mixing coefficients of $\{(\widehat{W}_{t}^\prime, h_{t}^\prime)^\prime\}$. This implies that $\widehat{Z}_{n,t}$ inherits the $\alpha$-mixing rate of $\{(\widehat{W}_{t}^\prime, h_{t}^\prime)^\prime\}$. Using similar arguments, we may also show that the $\widehat{Z}_{n,t}$ inherit the $\phi$-mixing rate from $\{(\widehat{W}_{t}^\prime, h_{t}^\prime)^\prime\}$. In both cases, $\widehat{Z}_{n,t}$ is mixing of the same size as $\{(\widehat{W}_{t}^\prime, h_{t}^\prime)^\prime\}$.

We now verify \eqref{eq:Conv}. Let $\lambda\in\mathbb{R}^{q}$ such that $\lambda^\prime\lambda=1$. Write
\begin{align*}
\lambda^{\prime}\Omega^{-1/2}\frac{1}{\sqrt{n}}\sum_{t=1}^{n}\widehat{Z}_{n,t} &= \lambda^{\prime}\Omega^{-1/2}\frac{1}{\sqrt{n}}\sum_{t=1}^{n}\left[\widehat{Z}_{n,t}-\frac{\delta}{\sqrt{n}}\right] + \lambda^{\prime}\Omega^{-1/2}\frac{1}{\sqrt{n}}\sum_{t=1}^{n}\frac{\delta}{\sqrt{n}}\\
&=\lambda^{\prime}\Omega^{-1/2}\frac{1}{\sqrt{n}}\sum_{t=1}^{n}\left[\widehat{Z}_{n,t}-\frac{\delta}{\sqrt{n}}\right] + \lambda^{\prime}\Omega^{-1/2}\delta.
\end{align*}
Next, we show that $\big\{\lambda^{\prime}\Omega^{-1/2}\big(\widehat{Z}_{n,t}-\frac{\delta}{\sqrt{n}}\big)\big\}_{n,t}$ is a zero-mean triangular array satisfying the conditions of Theorem~5.20 of \citet{Whi01}. The fact that it has mean zero follows from \eqref{eq:(p.8)}. Using D3, we have that
\begin{align*}
	\overline{\sigma}_n^2 &:= \Var\Big[\lambda^\prime\Omega^{-1/2}\frac{1}{\sqrt{n}}\sum_{t=1}^{n}\Big(\widehat{Z}_{n,t}-\frac{\delta}{\sqrt{n}}\Big)\Big]\\
	&= \lambda^\prime\Omega^{-1/2}\Var\Big[\frac{1}{\sqrt{n}}\sum_{t=1}^{n}\widehat{Z}_{n,t}\Big]\Omega^{-1/2}\lambda\\
	&=\lambda^\prime\Omega^{-1/2}\Omega_n\Omega^{-1/2}\lambda\underset{(n\to\infty)}{\longrightarrow}\lambda^\prime\lambda=1>0.
\end{align*}
Furthermore, by D2 there exists some $\Delta>0$, such that
\[
	\E\Big|\lambda^{\prime}\Omega^{-1/2}\Big(\widehat{Z}_{n,t}-\frac{\delta}{\sqrt{n}}\Big)\Big|^{2r}\leq\Delta<\infty.
\]
As pointed out above, $\widehat{Z}_{n,t}$ is mixing of the same size as $\{(W_{t}^\prime, h_{t}^\prime)^\prime\}$, so---in particular---it is $\alpha$-mixing of size $-2r/(2r-2)$ (since $-2r/(2r-2)$ implies weaker dependence than the actual $\alpha$-mixing rate of $-2r/(r-2)$) or $\phi$-mixing of size $-2r/[2(2r-2)]$ (since $-2r/[2(2r-2)]$ implies weaker dependence than the actual $\phi$-mixing rate of $-r/(r-1)$). Thus, all conditions of Theorem~5.20 in \citet{Whi01} are satisfied, and we conclude
\[
	\lambda^{\prime}\Omega^{-1/2}\frac{1}{\sqrt{n}}\sum_{t=1}^{n}\Big(\widehat{Z}_{n,t}-\frac{\delta}{\sqrt{n}}\Big)\overset{d}{\longrightarrow}N(0,1),\qquad\text{as }n\to\infty. 
\]
Since this holds for all $\lambda\in\mathbb{R}^{q}$ with $\lambda^\prime\lambda=1$, we get from a Cram\'{e}r--Wold device \citep[Proposition~5.1]{Whi01} that 
\[
	\Omega^{-1/2}\frac{1}{\sqrt{n}}\sum_{t=1}^{n}\Big(\widehat{Z}_{n,t}-\frac{\delta}{\sqrt{n}}\Big)\overset{d}{\longrightarrow}N(0,I_{q\times q}),\qquad\text{as }n\to\infty,
\]
or, equivalently, \eqref{eq:Conv}, that is,
\[
	\Omega^{-1/2}\frac{1}{\sqrt{n}}\sum_{t=1}^{n}\widehat{Z}_{n,t}\overset{d}{\longrightarrow}N(\Omega^{-1/2}\delta,I_{q\times q}),\qquad\text{as }n\to\infty.
\]
By assumption, $\widehat{\Omega}_n-\Omega_n=o_{\p}(1)$ and $\Omega_n-\Omega=o(1)$, whence $\widehat{\Omega}_n\overset{\p}{\longrightarrow}\Omega$. Due to this,
\[
	\widehat{\Omega}_n^{-1/2}\frac{1}{\sqrt{n}}\sum_{t=1}^{n}\widehat{Z}_{n,t}\overset{d}{\longrightarrow}N(\Omega^{-1/2}\delta,I_{q\times q}),\qquad\text{as }n\to\infty.
\]
The continuous mapping theorem implies that
\[
	\widehat{T}_n=n\left(\frac{1}{n}\sum_{t=1}^{n}\widehat{Z}_{n,t}\right)^\prime\widehat{\Omega}_n^{-1}\left(\frac{1}{n}\sum_{t=1}^{n}\widehat{Z}_{n,t}\right)\overset{d}{\longrightarrow}\chi_q^2(\delta^\prime\Omega^{-1}\delta),\qquad\text{as }n\to\infty.
\]
This concludes the proof.
\end{proof}

\begin{proof}[{\textbf{Proof of Proposition~\ref{prop:LRV}:}}]
By arguments similar to those used in the proof of Theorem~\ref{thm:2}, $\{\widehat{Z}_{n,t}\widehat{Z}_{n,t}^\prime\}$ is also mixing of the same size as $\{(\widehat{W}_{t}^\prime, h_{t}^\prime)^\prime\}$. Thus, $\{\widehat{Z}_{n,t}\widehat{Z}_{n,t}^\prime\}$ is $\alpha$-mixing of size $-2r/(r-2)$ or $\phi$-mixing of size $-r/(r-1)$. By the Cauchy--Schwarz inequality and D2,
\[
	\E|\widehat{Z}_{n,t}^{(i)}\widehat{Z}_{n,t}^{(j)}|^{r}\leq\sqrt{\E|\widehat{Z}_{n,t}^{(i)}|^{2r}}\sqrt{\E|\widehat{Z}_{n,t}^{(j)}|^{2r}}\leq\Delta_Z^{1/2}\Delta_Z^{1/2}<\infty.
\]
Hence, Theorem~6.20 in \citet{Whi01} shows for
\begin{align*}
\overline{\Omega}_n &:=\frac{1}{n}\sum_{t=1}^{n}\Big(\widehat{Z}_{n,t}-\frac{\delta}{\sqrt{n}}\Big)\Big(\widehat{Z}_{n,t}-\frac{\delta}{\sqrt{n}}\Big)^\prime\\
	&\hspace{1cm}+\frac{1}{n}\sum_{h=1}^{m_n}w_{n,h}\sum_{t=h+1}^{n}\Bigg\{\Big(\widehat{Z}_{n,t}-\frac{\delta}{\sqrt{n}}\Big)\Big(\widehat{Z}_{n,t-h}-\frac{\delta}{\sqrt{n}}\Big)^\prime + 
	\Big(\widehat{Z}_{n,t-h}-\frac{\delta}{\sqrt{n}}\Big)\Big(\widehat{Z}_{n,t}-\frac{\delta}{\sqrt{n}}\Big)^\prime\Bigg\}
\end{align*}
that
\begin{equation}\label{eq:(24.1)}
\overline{\Omega}_n-\Omega_n\overset{p}{\longrightarrow}0.
\end{equation}
Use that $|w_{n,h}|\leq\Delta_w<\infty$ and $m_n=o(n^{1/4})$ to decompose
\begin{align*}
	\overline{\Omega}_n&=\widehat{\Omega}_n-\frac{\delta}{n}\Bigg\{\frac{1}{n}\sum_{t=1}^{n}\widehat{Z}_{n,t}^\prime + \frac{1}{n}\sum_{h=1}^{m_n}w_{n,h}\sum_{t=h+1}^{n}\widehat{Z}_{n,t-h}^\prime + \frac{1}{n}\sum_{h=1}^{m_n}w_{n,h}\sum_{t=h+1}^{n}\widehat{Z}_{n,t}^\prime\Bigg\}\\
	&\hspace{1.3cm}-\Bigg\{\frac{1}{n}\sum_{t=1}^{n}\widehat{Z}_{n,t}+\frac{1}{n}\sum_{h=1}^{m_n}w_{n,h}\sum_{t=h+1}^{n}\widehat{Z}_{n,t-h} + \frac{1}{n}\sum_{h=1}^{m_n}w_{n,h}\sum_{t=h+1}^{n}\widehat{Z}_{n,t}\Bigg\}\frac{\delta^\prime}{\sqrt{n}}\\
	&\hspace{1.3cm}+\Bigg\{\frac{2}{n}\sum_{h=1}^{m_n}w_{n,h}\sum_{t=h+1}^{n}\frac{\delta\delta^\prime}{n}\Bigg\}+\frac{\delta\delta^\prime}{n}\\
	&=\widehat{\Omega}_n - \frac{\delta}{n}\Bigg\{\frac{1}{n}\sum_{t=1}^{n}\Big(\widehat{Z}_{n,t}^\prime-\frac{\delta^\prime}{\sqrt{n}}\Big) + \frac{\delta^\prime}{\sqrt{n}} \\
	&\hspace{3cm}+ \frac{1}{n}\sum_{h=1}^{m_n}w_{n,h}\sum_{t=h+1}^{n}\Big(\widehat{Z}_{n,t-h}^\prime-\frac{\delta^\prime}{\sqrt{n}}\Big)+\frac{1}{n}\sum_{h=1}^{m_n}w_{n,h}(n-h)\frac{\delta^\prime}{\sqrt{n}} \\
	&\hspace{3cm}+ \frac{1}{n}\sum_{h=1}^{m_n}w_{n,h}\sum_{t=h+1}^{n}\Big(\widehat{Z}_{n,t}^\prime-\frac{\delta^\prime}{\sqrt{n}}\Big)+\frac{1}{n}\sum_{h=1}^{m_n}w_{n,h}(n-h)\frac{\delta^\prime}{\sqrt{n}}\Bigg\}\\
	&\hspace{1.4cm}- \Bigg\{\frac{1}{n}\sum_{t=1}^{n}\Big(\widehat{Z}_{n,t}-\frac{\delta}{\sqrt{n}}\Big) + \frac{\delta}{\sqrt{n}}  \\
	&\hspace{3cm}+ \frac{1}{n}\sum_{h=1}^{m_n}w_{n,h}\sum_{t=h+1}^{n}\Big(\widehat{Z}_{n,t-h}-\frac{\delta}{\sqrt{n}}\Big)+\frac{1}{n}\sum_{h=1}^{m_n}w_{n,h}(n-h)\frac{\delta}{\sqrt{n}} \\
	&\hspace{3cm}+ \frac{1}{n}\sum_{h=1}^{m_n}w_{n,h}\sum_{t=h+1}^{n}\Big(\widehat{Z}_{n,t}-\frac{\delta}{\sqrt{n}}\Big)+\frac{1}{n}\sum_{h=1}^{m_n}w_{n,h}(n-h)\frac{\delta}{\sqrt{n}}\Bigg\}\frac{\delta^\prime}{n}+o(1)\\
	&=\widehat{\Omega}_n -\frac{\delta}{\sqrt{n}}o_{\p}(1) - o_{\p}(1)\frac{\delta^\prime}{\sqrt{n}}+o(1)\\
	&=\widehat{\Omega}_n + o_{\p}(1),
\end{align*}
where we have also used Lemma~6.19 of \citet{Whi01} in the second to last step. Together with \eqref{eq:(24.1)} this implies that $\widehat{\Omega}_n-\Omega_n=o_{\p}(1)$, as desired.
\end{proof}

\begin{proof}[{\textbf{Proof of Proposition~\ref{prop:VarYt}:}}]
First, use \eqref{eq:L char} to write
\begin{equation}\label{eq:decomp Breg}
	d(\widehat Y_t, x_{1t}, x_{2t}) =	a_{t-1} + b_{t-1}^\prime \widehat Y_t,
\end{equation}
where the $\mathcal{F}_{t-1}$-measurable $a_{t-1}$ and $b_{t-1}$ do not depend on $\widehat Y_t$ and are of the form
\begin{align*}
	a_{t-1} &= \phi(x_{1t}) - \phi(x_{2t}) - \langle x_{1t}, \D\phi(x_{1t})\rangle + \langle x_{2t}, \D\phi(x_{2t})\rangle, \\
	b_{t-1} &= \D\phi(x_{1t}) - \D\phi(x_{2t}).
\end{align*}

	For the first statement of the proposition, use \eqref{eq:decomp Breg} to obtain
	\begin{align}
\Omega_n 
	&= \Var \Big( n^{-1/2} \sum_{t=1}^n h_{t-1} d\big(\widehat{Y}_t, x_{1t}, x_{2t}\big) \Big)
	= \Var \Big( n^{-1/2} \sum_{t=1}^n h_{t-1} \big( a_{t-1} + b_{t-1}^\prime \widehat Y_t \big) \Big) \notag\\
	&= \frac{1}{n} \sum_{t=1}^n \Var \Big(  h_{t-1} \big( a_{t-1} + b_{t-1}^\prime \widehat Y_t \big) \Big) \notag\\
	&\qquad + \frac{2}{n} \sum_{s < t} \Cov \Big( h_{s-1} \big( a_{s-1} + b_{s-1}^\prime \widehat Y_s \big) ,\; h_{t-1} \big( a_{t-1} + b_{t-1}^\prime \widehat Y_t \big)  \Big).\label{eq:Omega decomp}
	\end{align}
Now, exploit that $x_t^\ast = \E_{t-1} [Y_t]$ and the law of total (co)variance to get
\begin{align}
&\Var \Big(  h_{t-1} \big( a_{t-1} + b_{t-1}^\prime \widehat Y_t \big) \Big)\notag\\
&=\Var\Big(\E_{t-1}\Big[h_{t-1} \big( a_{t-1} + b_{t-1}^\prime \widehat Y_t \big)\Big]\Big) + \E\Big[\Var_{t-1}\Big(  h_{t-1} \big( a_{t-1} + b_{t-1}^\prime \widehat Y_t \big) \Big)\Big]\notag\\
&=\Var\Big(\E_{t-1}\Big[h_{t-1} \big( a_{t-1} + b_{t-1}^\prime x_t^\ast \big) + h_{t-1}b_{t-1}^\prime(\widehat{Y}_t- x_t^\ast)\Big]\Big) + \E\Big[h_{t-1}b_{t-1}^\prime\Var_{t-1}\Big( \widehat Y_t \big)b_{t-1} h_{t-1}^\prime \Big)\Big]\notag\\
&=\Var \Big(  h_{t-1} \big( a_{t-1} + b_{t-1}^\prime x_t^{\ast}\big) \Big) + \E\Big[  h_{t-1}b_{t-1}^\prime \Var_{t-1}\big(\widehat{Y}_t\big)b_{t-1}h_{t-1}^\prime \Big]\label{eq:decomp Var}
\end{align}
and, additionally exploiting $\Cov_{t-1} \big( \widehat Y_s ,\;   \widehat Y_t \big) = 0$ and $\widehat Y_s = \E_{t-1} [\widehat Y_s]$ for $s<t$,
\begin{align}
\Cov &\Big( h_{s-1} \big( a_{s-1} + b_{s-1}^\prime \widehat Y_s \big) ,\; h_{t-1} \big( a_{t-1} + b_{t-1}^\prime \widehat Y_t \big)  \Big)\notag\\
&= \Cov \Big( h_{s-1} \big( a_{s-1} + b_{s-1}^\prime \widehat Y_s \big) ,\; h_{t-1} \big( a_{t-1} + b_{t-1}^\prime \{\widehat Y_t- \E_{t-1}[\widehat Y_t]\}\big) + h_{t-1}b_{t-1}^\prime\E_{t-1}[\widehat Y_t]   \Big)\notag\\
&= \Cov \Big( h_{s-1} \big( a_{s-1} + b_{s-1}^\prime \widehat Y_s \big) ,\; h_{t-1} \big( a_{t-1} + b_{t-1}^\prime \{\widehat Y_t- \E_{t-1}[\widehat Y_t]\}\big)\Big)\notag\\
&\qquad +\Cov\Big( h_{s-1} \big( a_{s-1} + b_{s-1}^\prime \widehat Y_s \big) ,\;  h_{t-1}b_{t-1}^\prime\E_{t-1}[\widehat Y_t] \big)  \Big)\notag\\
&= \Cov \Big( h_{s-1} \big( a_{s-1} + b_{s-1}^\prime \widehat Y_s \big) ,\; h_{t-1} a_{t-1} \Big)\notag\\
&\qquad +\Cov\Big( h_{s-1} \big( a_{s-1} + b_{s-1}^\prime \widehat Y_s \big) ,\;  h_{t-1}b_{t-1}^\prime x_t^\ast  \Big)\notag\\
&= \Cov \Big( h_{s-1} d(\widehat Y_s,x_{1s}, x_{2s}) ,\; h_{t-1} d(x_t^\ast,x_{1t}, x_{2t}) \Big).\label{eq:decomp Covar}
\end{align}
Plugging \eqref{eq:decomp Var} and \eqref{eq:decomp Covar} into \eqref{eq:Omega decomp} gives
	\begin{align*}
\Omega_n
	&= \frac{1}{n} \sum_{t=1}^n \Var \Big( h_{t-1} \big( a_{t-1} + b_{t-1}^\prime x_t^\ast \big) \Big) 
	+ \frac{1}{n} \sum_{t=1}^n \E \Big[  h_{t-1}  b_{t-1}^\prime \Var_{t-1} \big( \widehat Y_t \big)b_{t-1} h_{t-1}^\prime \Big] \\
	&\qquad+ \frac{2}{n} \sum_{s < t} \Cov \Big( h_{s-1} \big( a_{s-1} + b_{s-1}^\prime  \widehat Y_s \big) ,\; h_{t-1} \big(  a_{t-1} + b_{t-1}^\prime  x_t^\ast \big) \Big).
	\end{align*}
	
	The next statement of the proposition in \eqref{eqn:ProxyVariance} is immediate.
	Hence, it remains to show the lower bound, i.e., that $\frac{1}{n} \sum_{t=1}^n \E \left[h_{t-1}    b_{t-1}^\prime \Var_{t-1} \big( \widehat Y_t \big)b_{t-1} h_{t-1}^\prime\right]$ is positive semi-definite.
	For any $v \in \mathbb{R}^q$, $v \not= 0$, it holds that
	\begin{align*}
	v^\prime \Bigg( \frac{1}{n} \sum_{t=1}^n \E \big[  h_{t-1} h_{t-1}^\prime b_{t-1}^\prime \Var_{t-1} \big( \widehat Y_t \big) b_{t-1} \big] \Bigg) v 
	&= \frac{1}{n} \sum_{t=1}^n \E \Big[ \big(v^\prime h_{t-1}\big) \Var_{t-1} \big(b_{t-1}^\prime \widehat Y_t \big) \big(h_{t-1}^\prime v)\big)  \Big]  \\
	& =\frac{1}{n} \sum_{t=1}^n \E \Big[ \Var_{t-1} \Big(  \big(v^\prime h_{t-1}\big)  b_{t-1} \widehat Y_t \Big) \Big] \ge 0,
	\end{align*}
	which concludes this proof.
\end{proof}

\begin{proof}[{\textbf{Proof of Proposition~\ref{prop:QuantileRobustness}:}}]
	In this proof, we show the stronger statement that if for any $\omega \in \Omega$,
	\begin{align*}
	\E_{t-1} \left[ d_{g,\alpha}({Y}_t, x_{1t}, x_{2t}) -  d_{g,\alpha}(\widehat{Y}_t, x_{1t}, x_{2t}) \right] (\omega) = 0 \; \text{ for all  $x_{1t}(\omega), x_{2t}(\omega) \in \mathsf{A}$},
	\end{align*}
	then, $F_t(\omega, \cdot) = \widehat F_t(\omega, \cdot)$.
	
	We prove this statement by contradiction.
	For this, let $\omega \in \Omega$ and assume that there exists $z = z(\omega) \in \mathbb{R}$ such that $F_t(\omega, z) < \widehat F_t(\omega, z)$.
	(The proof for the case that $F_t(\omega, z) > \widehat F_t(\omega, z)$ proceeds analogously by interchanging $F_t$ and $\widehat F_t$.)
	Then, we set $x_{1t} \equiv x_{1t}(\omega) \in \mathsf{A}$ as the largest point smaller than $z$ such that $F_t \big(\omega, x_{1t} \big) = \widehat F_t \big( \omega, x_{1t} \big)$, and $x_{2t} \equiv x_{2t}(\omega) \in \mathsf{A}$ as the smallest point larger than $z$ such that $F_t \big(\omega, x_{2t} \big) = \widehat F_t \big( \omega, x_{2t} \big)$. 
	Formally,
	\begin{align*}
	x_{1t}(\omega) &:= \sup \left\{ x \in \mathsf{A}, x \le z : \; F_t \big(\omega, x \big) = \widehat F_t \big( \omega, x \big) \right\},
	\qquad \text{ and } \\
	x_{2t}(\omega) &:= \inf \left\{ x \in \mathsf{A}, x \ge z : \; F_t \big(\omega, x \big) = \widehat F_t \big( \omega, x \big) \right\},
	\end{align*}
	which are well-defined as the conditional distributions $F_t \big(\omega, \cdot \big)$ and  $\widehat F_t \big( \omega, \cdot \big)$ are continuous by assumption and have bounded support.
	Hence, \eqref{eqn:QuantLossDiffDiff} simplifies to
	\begin{multline*}
	\E_{t-1} \left[ d_{g,\alpha}({Y}_t, x_{1t}, x_{2t}) -  d_{g,\alpha}(\widehat{Y}_t, x_{1t}, x_{2t}) \right] \\
	=  \E_{t-1} \left[ g(Y_t)  \mathds{1}(x_{1t} < Y_t \le x_{2t})  - g(\widehat Y_t)  \mathds{1}(x_{1t} < \widehat Y_t \le x_{2t}) \right].
	\end{multline*}
	
	We now consider the random variables $Y_{R,t}$ and $\widehat Y_{R,t}$, which are versions of the random variables $Y_t$ and $\widehat Y_t$, restricted to the interval $[x_{1t} , x_{2t}]$, with conditional distribution functions
	\begin{align*}
	F_{R,t}(\omega,y) = \frac{F_t(\omega,y) - F_t(\omega,x_{1t})}{F_t(\omega,x_{2t}) - F_t(\omega,x_{1t})} 
	\quad \text{ and } \quad
	\widehat F_{R,t}(\omega,y) = \frac{\widehat F_t(\omega,y) - \widehat F_t(\omega,x_{1t})}{\widehat F_t(\omega,x_{2t}) - \widehat F_t(\omega,x_{1t})}
	\end{align*}
	for all $y \in [x_{1t}, x_{2t}]$.
	(Notice that $ F_t(\omega, x_{2t}) -  F_t(\omega, x_{1t}) > 0$ and $\widehat F_t(\omega, x_{2t}) - \widehat F_t(\omega, x_{1t}) > 0$ follows from the conditions above.)
	Then, the conditional distribution of $\widehat Y_{R,t}$ is smaller than the one of $Y_{R,t}$ in the usual stochastic order \citep{shaked2007stochastic}. Invoking \citet[Theorem 1.A.8]{shaked2007stochastic} implies that  if
	\begin{align*}
	\E_{t-1} \left[ g(Y_t) \mathds{1}(x_{1t} < Y_t \le x_{2t})   \right]  
	= \E_{t-1} \left[   g(Y_{R,t}) \right]  
	\stackrel{!}{=} \E_{t-1} \left[   g(\widehat Y_{R,t}) \right]  
	= \E_{t-1} \left[ g(\widehat Y_t) \mathds{1}(x_{1t} < \widehat Y_t \le x_{2t})  \right]
	\end{align*}
	holds for \textit{some} strictly increasing function $g(\cdot)$, this already implies equality of the conditional distributions of $Y_{R,t}$ and $\widehat Y_{R,t}$.
	Thus, restricted to the interval $[x_{1t} , x_{2t}]$,  the conditional distributions of $Y_t$ and $\widehat Y_t$ coincide, which contradicts the presumed existence of $z = z(\omega) \in \mathbb{R}$ such that $F_t(\omega, z) < \widehat F_t(\omega, z)$. This completes the proof.
\end{proof}

\section{Technical Derivations for Section~\ref{Simulations for robust loss functions}}
\label{Technical Derivations}

\renewcommand{\theequation}{C.\arabic{equation}}	
\setcounter{equation}{0}

This appendix collects the proof of equation~\eqref{eq:Sim1} in the main paper. It also derives the formula for the variance-covariance matrix $\Omega$ from Assumption~D3 in Proposition~\ref{prop:Omega}.

\begin{proof}[{\textbf{Proof of Equation~\eqref{eq:Sim1}:}}]
Straightforward calculations exploiting \eqref{eq:1}--\eqref{eq:4} yield that
\begin{align}
	d(Y_t, x_{1t}, x_{2t})&=(Y_t - x_{1t})^2 - (Y_t - x_{2t})^2 
	=(\varepsilon_t-\varepsilon_{1,t-1})^2-(\mu(1-\phi)+\varepsilon_t)^2\notag\\
	&=-2\varepsilon_t\varepsilon_{1,t-1}+\varepsilon_{1,t-1}^2-[\mu^2(1-\phi)^2 + 2\mu(1-\phi)\varepsilon_t].\label{eq:(11m)}
\end{align}
Using the mutual independence of $\{\varepsilon_t\}$ and $\{\varepsilon_{1t}\}$, we get
\begin{equation*}
	\E[d(Y_t, x_{1t}, x_{2t})] = \sigma_1^2 - \mu^2(1-\phi)^2 = \xi/\sqrt{n}.
\end{equation*}
Moreover, $\widehat{Y}_{t-1}$ and $d(Y_t, x_{1t}, x_{2t})$ are independent by \eqref{eq:(11m)}, implying
\begin{align*}
	\E[\widehat{Y}_{t-1}d(Y_t, x_{1t}, x_{2t})] &= \E[\widehat{Y}_{t-1}]\E[d(Y_t, x_{1t}, x_{2t})]\\
	&= \big\{\E[Y_{t-1}]+\E[\widehat{\varepsilon}_{t-1}]\big\}\E[d(Y_t, x_{1t}, x_{2t})]\\
	&= \mu\xi/\sqrt{n}.
\end{align*}
This concludes the proof.
\end{proof}

\begin{prop}\label{prop:Omega}
Under the conditions of Section~\ref{Simulations for robust loss functions} in the main paper, it holds that
\begin{equation}\label{eq:Omega}
	\Omega = c_{\mu,\phi,\sigma_{\varepsilon}^2,\sigma_{\widehat{\varepsilon}}^2}\begin{pmatrix}1 & \mu\\
												\mu & \frac{\sigma_{\varepsilon}^2}{1-\phi^2}+\mu^2+\sigma_{\widehat{\varepsilon}}^2\end{pmatrix},
\end{equation}
where $c_{\mu,\phi,\sigma_{\varepsilon}^2,\sigma_{\widehat{\varepsilon}}^2} =2\mu^4(1-\phi)^4 + 8\mu^2(1-\phi)^2 (\sigma_{\varepsilon}^2 + \sigma_{\widehat{\varepsilon}}^2)$.
\end{prop}

\begin{proof}
Decompose
\begin{align*}
\Omega_n &= \Var\Big(\frac{1}{\sqrt{n}}\sum_{t=1}^{n}\widehat{Z}_{n,t}\Big)\\
	&= \frac{1}{n}\sum_{t=1}^{n}\E\Big[\Big(\widehat{Z}_{n,t}-\frac{\delta}{\sqrt{n}}\Big)\Big(\widehat{Z}_{n,t}-\frac{\delta}{\sqrt{n}}\Big)^\prime\Big]\\
	& \hspace{1.3cm}+\frac{1}{n}\sum_{h=1}^{n}\sum_{t=h+1}^{n}\E\Big[\Big(\widehat{Z}_{n,t}-\frac{\delta}{\sqrt{n}}\Big)\Big(\widehat{Z}_{n,t-h}-\frac{\delta}{\sqrt{n}}\Big)^\prime\Big]\\
	& \hspace{1.3cm}+\frac{1}{n}\sum_{h=1}^{n}\sum_{t=h+1}^{n}\E\Big[\Big(\widehat{Z}_{n,t-h}-\frac{\delta}{\sqrt{n}}\Big)\Big(\widehat{Z}_{n,t}-\frac{\delta}{\sqrt{n}}\Big)^\prime\Big]\\
	&=(I) + (II) + (III).
\end{align*}
Consider the terms $(I)$, $(II)$, and $(III)$ separately. For $(I)$, we obtain under $H_{a,\loc}$ that
\begin{align*}
	(I) &= \frac{1}{n}\sum_{t=1}^{n}\Big\{\E\big[\widehat{Z}_{n,t}\widehat{Z}_{n,t}^\prime\big] - \E\big[\widehat{Z}_{n,t}\big]\frac{\delta^\prime}{\sqrt{n}} - \frac{\delta}{\sqrt{n}}\E\big[\widehat{Z}_{n,t}^\prime\big]+\frac{\delta\delta^\prime}{n}\Big\}\\
	&= \Big\{\frac{1}{n}\sum_{t=1}^{n}\E\big[\widehat{Z}_{n,t}\widehat{Z}_{n,t}^\prime\big]\Big\} - \frac{\delta\delta^\prime}{n}\\
	&= \Big\{\frac{1}{n}\sum_{t=1}^{n}\E\big[\widehat{Z}_{n,t}\widehat{Z}_{n,t}^\prime\big]\Big\} +o(1).
\end{align*}
Recall that
\[
	\widehat{Z}_{n,t}=h_{t-1}d(\widehat{Y}_t, x_{1t}, x_{2t}) = \begin{pmatrix} d(\widehat{Y}_t, x_{1t}, x_{2t}) \\ \widehat{Y}_{t-1}d(\widehat{Y}_t, x_{1t}, x_{2t})\end{pmatrix}.
\]
Then, 
\[
	\E[\widehat{Z}_{n,t}\widehat{Z}_{n,t}^\prime]=\E\left[\begin{pmatrix} d^2(\widehat{Y}_t, x_{1t}, x_{2t}) & \widehat{Y}_{t-1}d^2(\widehat{Y}_t, x_{1t}, x_{2t}) \\ 
	\widehat{Y}_{t-1}d^2(\widehat{Y}_t, x_{1t}, x_{2t}) & \widehat{Y}_{t-1}^2d^2(\widehat{Y}_t, x_{1t}, x_{2t})\end{pmatrix}\right].
\]
Consider each matrix element separately. Straightforward calculations yield that
\begin{equation}\label{eq:(p.3)}
	d(\widehat{Y}_t, x_{1t}, x_{2t}) = (\varepsilon_t+\widehat{\varepsilon}_t-\varepsilon_{1,t-1})^2 - (\mu(1-\phi) + \varepsilon_t + \widehat{\varepsilon}_t)^2.
\end{equation}
Hence, exploiting the independence of the disturbances,
\begin{align}
	\E[d^2(\widehat{Y}_t, x_{1t}, x_{2t})] &= \E[(\varepsilon_t+\widehat{\varepsilon}_t-\varepsilon_{1,t-1})^4] \notag\\
	&\hspace{2cm} - 2\E[(\varepsilon_t+\widehat{\varepsilon}_t-\varepsilon_{1,t-1})^2(\mu(1-\phi) + \varepsilon_t + \widehat{\varepsilon}_t)^2] \notag\\
	&\hspace{2cm} + \E[(\mu(1-\phi) + \varepsilon_t + \widehat{\varepsilon}_t)^4].\label{eq:d^2 prel}
\end{align}	
Since $(\varepsilon_t+\widehat{\varepsilon}_t-\varepsilon_{1,t-1})\sim N(0,\sigma_{\varepsilon}^2+\sigma_{\widehat{\varepsilon}}^2+\sigma_{1}^2)$ and $(\mu(1-\phi) + \varepsilon_t + \widehat{\varepsilon}_t)\sim N\big(\mu(1-\phi), \sigma_{\varepsilon}^2+\sigma_{\widehat{\varepsilon}}^2\big)$, we have
\begin{align}
	\E[(\varepsilon_t+\widehat{\varepsilon}_t-\varepsilon_{1,t-1})^4] &= 3(\sigma_{\varepsilon}^2+\sigma_{\widehat{\varepsilon}}^2+\sigma_{1}^2)^2,\label{eq:e1}\\
	\E[(\mu(1-\phi) + \varepsilon_t + \widehat{\varepsilon}_t)^4] &= \mu^4(1-\phi)^4 + 6\mu^2(1-\phi)^2(\sigma_{\varepsilon}^2+\sigma_{\widehat{\varepsilon}}^2) + 3(\sigma_{\varepsilon}^2+\sigma_{\widehat{\varepsilon}}^2)^2.\label{eq:e2}
\end{align}
Furthermore,
\begin{align}
& \E[(\varepsilon_t+\widehat{\varepsilon}_t-\varepsilon_{1,t-1})^2(\mu(1-\phi) + \varepsilon_t + \widehat{\varepsilon}_t)^2] \notag\\
&= \E\left[\left\{(\varepsilon_t+\widehat{\varepsilon}_t)^2 - 2\varepsilon_{1,t-1}(\varepsilon_t+\widehat{\varepsilon}_t) + \varepsilon_{1,t-1}^2\right\}\left\{\mu^2(1-\phi)^2+2\mu(1-\phi)(\varepsilon_t+\widehat{\varepsilon}_t)+(\varepsilon_t+\widehat{\varepsilon}_t)^2\right\}\right]\notag\\
&= \E\left[\left\{(\varepsilon_t+\widehat{\varepsilon}_t)^2 + \varepsilon_{1,t-1}^2\right\}\left\{\mu^2(1-\phi)^2+2\mu(1-\phi)(\varepsilon_t+\widehat{\varepsilon}_t)+(\varepsilon_t+\widehat{\varepsilon}_t)^2\right\}\right]\notag\\
&= \mu^2(1-\phi)^2\left\{\E[(\varepsilon_t+\widehat{\varepsilon}_t)^2]+\E[\varepsilon_{1,t-1}^2]\right\}\notag\\
&\hspace{2cm} +2\mu(1-\phi)\left\{\E[(\varepsilon_t+\widehat{\varepsilon}_t)^3]+\E[\varepsilon_{1,t-1}^2(\varepsilon_t+\widehat{\varepsilon}_t)]\right\}\notag\\
&\hspace{2cm} + \E[(\varepsilon_t+\widehat{\varepsilon}_t)^4]+ \E[(\varepsilon_t+\widehat{\varepsilon}_t)^2]\E[\varepsilon_{1,t-1}^2]\notag\\
&= \mu^2(1-\phi)^2(\sigma_{\varepsilon}^2 + \sigma_{\widehat{\varepsilon}}^2+\sigma_{1}^2) + 3(\sigma_{\varepsilon}^2 + \sigma_{\widehat{\varepsilon}}^2)^2+ (\sigma_{\varepsilon}^2 + \sigma_{\widehat{\varepsilon}}^2)\sigma_1^2.\label{eq:e3}
\end{align}
Inserting \eqref{eq:e1}--\eqref{eq:e3} into \eqref{eq:d^2 prel} gives
\begin{align}
	\E[d^2(\widehat{Y}_t, x_{1t}, x_{2t})] &= 3(\sigma_{\varepsilon}^2+\sigma_{\widehat{\varepsilon}}^2+\sigma_{1}^2)^2 - 2\mu^2(1-\phi)^2(\sigma_{\varepsilon}^2 + \sigma_{\widehat{\varepsilon}}^2+\sigma_{1}^2)-6(\sigma_{\varepsilon}^2 + \sigma_{\widehat{\varepsilon}}^2)^2\notag\\
	& \hspace{1cm}-2 (\sigma_{\varepsilon}^2 + \sigma_{\widehat{\varepsilon}}^2)\sigma_1^2 +\left[\mu^4(1-\phi)^4 + 6\mu^2(1-\phi)^2(\sigma_{\varepsilon}^2+\sigma_{\widehat{\varepsilon}}^2) + 3(\sigma_{\varepsilon}^2+\sigma_{\widehat{\varepsilon}}^2)^2\right]\notag\\
	&=:c_{\mu,\phi,\sigma_{\varepsilon}^2,\sigma_{\widehat{\varepsilon}}^2,\sigma_1^2}.\label{eq:d^2}
\end{align}

For the remaining terms in $\E[\widehat{Z}_{n,t}\widehat{Z}_{n,t}^\prime]$, note that $\widehat{Y}_{t-1}$ is independent of $d(\widehat{Y}_t, x_{1t}, x_{2t})$, which can easily be seen from \eqref{eq:(p.3)}. Therefore,
\begin{align}
	\E[\widehat{Y}_{t-1}d^2(\widehat{Y}_t, x_{1t}, x_{2t})] &= \E[\widehat{Y}_{t-1}]\E[d^2(\widehat{Y}_t, x_{1t}, x_{2t})] \notag\\
	&=  \mu \cdot c_{\mu,\phi,\sigma_{\varepsilon}^2,\sigma_{\widehat{\varepsilon}}^2,\sigma_1^2}\label{eq:Yd^2}
\end{align}
and
\[
	\E[\widehat{Y}_{t-1}^2d^2(\widehat{Y}_t, x_{1t}, x_{2t})] = \E[\widehat{Y}_{t-1}^2]\E[d^2(\widehat{Y}_t, x_{1t}, x_{2t})].
\]
We have that
\begin{align*}
	\E[\widehat{Y}_{t-1}^2] &= \E[Y_{t-1}+\widehat{\varepsilon}_{t-1}]^2\\
	&= \E[Y_{t-1}^2]+2\E[Y_{t-1}]\E[\widehat{\varepsilon}_{t-1}] + \E[\widehat{\varepsilon}_{t-1}^2]\\
	&= \E[Y_{t-1}^2] + 2\cdot\mu\cdot 0 + \sigma_{\widehat{\varepsilon}}^2\\
	&= \E[Y_{t-1}^2] + \sigma_{\widehat{\varepsilon}}^2.
\end{align*}
For our Gaussian AR(1), it is well-known that $Y_t\sim N(\mu, \sigma_{\varepsilon}^2/(1-\phi^2))$ and, therefore,
\begin{equation}\label{eq:(p.30)}
	\E[Y_{t-1}^2]=\Var(Y_{t-1})+\{\E[Y_{t-1}]\}^2 = \frac{\sigma_{\varepsilon}^2}{1-\phi^2}+\mu^2. 
\end{equation}
Hence, with \eqref{eq:d^2},
\begin{equation}\label{eq:Y^2d^2}
\E[\widehat{Y}_{t-1}^2d^2(\widehat{Y}_t, x_{1t}, x_{2t})] = \left(\frac{\sigma_{\varepsilon}^2}{1-\phi^2}+\mu^2+ \sigma_{\widehat{\varepsilon}}^2\right)c_{\mu,\phi,\sigma_{\varepsilon}^2,\sigma_{\widehat{\varepsilon}}^2,\sigma_1^2}.
\end{equation}
Combining \eqref{eq:d^2}--\eqref{eq:Y^2d^2}, we obtain that
\begin{equation*}
	\frac{1}{n}\sum_{t=1}^{n}\E[\widehat{Z}_{n,t}\widehat{Z}_{n,t}^\prime] = c_{\mu,\phi,\sigma_{\varepsilon}^2,\sigma_{\widehat{\varepsilon}}^2,\sigma_1^2}\begin{pmatrix}1 & \mu\\
												\mu & \frac{\sigma_{\varepsilon}^2}{1-\phi^2}+\mu^2+\sigma_{\widehat{\varepsilon}}^2\end{pmatrix}.
\end{equation*}
Since $\sigma_1^2=\mu^2(1-\phi)^2+\xi/\sqrt{n}\rightarrow\mu^2(1-\phi)^2$, as $n\to\infty$, we have $c_{\mu,\phi,\sigma_{\varepsilon}^2,\sigma_{\widehat{\varepsilon}}^2,\sigma_1^2}\rightarrow c_{\mu,\phi,\sigma_{\varepsilon}^2,\sigma_{\widehat{\varepsilon}}^2}$. Thus, $(I)\rightarrow\Omega$, as $n\to\infty$.

It remains to show that $(II) + (III)=o(1)$. As $(III)=(II)^\prime$, we only have to show that $(II)=o(1)$. Write $\delta=(\delta_1,\delta_2)^\prime$ and consider the four elements of $\E\Big[\Big(\widehat{Z}_{n,t}-\frac{\delta}{\sqrt{n}}\Big)\Big(\widehat{Z}_{n,t-h}-\frac{\delta}{\sqrt{n}}\Big)^\prime\Big]$ separately. For the upper-left element, we have that
\begin{align*}
\E&\Big[\Big(d(\widehat{Y}_{t},x_{1t}, x_{2t})-\frac{\delta_1}{\sqrt{n}}\Big)\Big(d(\widehat{Y}_{t-h},x_{1t-h}, x_{2t-h})-\frac{\delta_1}{\sqrt{n}}\Big)\Big] \\
&=\E\Big[\Big(d(\widehat{Y}_{t},x_{1t}, x_{2t})-\frac{\delta_1}{\sqrt{n}}\Big)\Big]\E\Big[\Big(d(\widehat{Y}_{t-h},x_{1t-h}, x_{2t-h})-\frac{\delta_1}{\sqrt{n}}\Big)\Big]\\
&=0 \cdot0=0 
\end{align*}
by \eqref{eq:(p.3)} and independence of the disturbances. By identical arguments, we also obtain for the upper-right element that
\begin{align*}
	\E&\Big[\Big(d(\widehat{Y}_{t},x_{1t}, x_{2t})-\frac{\delta_1}{\sqrt{n}}\Big)\Big(\widehat{Y}_{t-h-1}d(\widehat{Y}_{t-h},x_{1t-h}, x_{2t-h})-\frac{\delta_2}{\sqrt{n}}\Big)\Big] \\
	&=\E\Big[\Big(d(\widehat{Y}_{t},x_{1t}, x_{2t})-\frac{\delta_1}{\sqrt{n}}\Big)\Big]\E\Big[\Big(\widehat{Y}_{t-h-1}d(\widehat{Y}_{t-h},x_{1t-h}, x_{2t-h})-\frac{\delta_2}{\sqrt{n}}\Big)\Big] \\
	&= 0\cdot 0 =0.
\end{align*}

For the lower-left element, we get that
\begin{align}
\E&\Big[\Big(\widehat{Y}_{t-1}d(\widehat{Y}_{t},x_{1t}, x_{2t})-\frac{\delta_2}{\sqrt{n}}\Big)\Big(d(\widehat{Y}_{t-h},x_{1t-h}, x_{2t-h})-\frac{\delta_1}{\sqrt{n}}\Big)\Big]\notag\\
&= \E\Big[\widehat{Y}_{t-1}d(\widehat{Y}_{t},x_{1t}, x_{2t})d(\widehat{Y}_{t-h},x_{1t-h}, x_{2t-h})\Big] - \frac{\delta_1}{\sqrt{n}}\E\Big[\widehat{Y}_{t-1}d(\widehat{Y}_{t},x_{1t}, x_{2t})\Big]\notag\\
&\hspace{7.5cm} - \frac{\delta_2}{\sqrt{n}}\E\Big[d(\widehat{Y}_{t-h},x_{1t-h}, x_{2t-h})\Big] + \frac{\delta_1\delta_2}{n}\notag\\
&=\E\Big[d(\widehat{Y}_{t},x_{1t}, x_{2t})\Big]\E\Big[\widehat{Y}_{t-1}d(\widehat{Y}_{t-h},x_{1t-h}, x_{2t-h})\Big] - \frac{\delta_1\delta_2}{n} - \frac{\delta_2\delta_1}{n} + \frac{\delta_1\delta_2}{n}\notag\\
&= \frac{\delta_1}{\sqrt{n}}\E\Big[\widehat{Y}_{t-1}d(\widehat{Y}_{t-h},x_{1t-h}, x_{2t-h})\Big]-\frac{\delta_1\delta_2}{n}.\label{eq:Yd}
\end{align}
Exploit \eqref{eq:1} and \eqref{eq:2} to deduce that
\begin{align}
	\widehat{Y}_{t-1}&=Y_{t-1}+\widehat{\varepsilon}_{t-1}\notag\\
	&= \mu(1-\phi)\Big(\sum_{i=1}^{h}\phi^{i-1}\Big) + \phi^{h}Y_{t-h-1} + \phi^{h}\varepsilon_{t-h} + \sum_{i=1}^{h-1}\phi^{i}\varepsilon_{t-i} +\widehat{\varepsilon}_{t-1}.\label{eq:Yhat decomp}
\end{align}
Hence, we obtain for $h>1$ that
\begin{align}
\E&\Big[\widehat{Y}_{t-1}d(\widehat{Y}_{t-h},x_{1t-h}, x_{2t-h})\Big]\notag\\
&=\E\Big[\Big\{\mu(1-\phi^h) + \phi^{h}Y_{t-h-1}+ \phi^{h}\varepsilon_{t-h} + \sum_{i=1}^{h-1}\phi^{i}\varepsilon_{t-i} +\widehat{\varepsilon}_{t-1}\Big\}d(\widehat{Y}_{t-h},x_{1t-h}, x_{2t-h})\Big]\label{eq:h=1}\\
&=\mu(1-\phi^h)\frac{\delta_1}{\sqrt{n}}+\phi^{h}\frac{\delta_2}{\sqrt{n}}+\phi^{h}\E\Big[\varepsilon_{t-h}d(\widehat{Y}_{t-h},x_{1t-h}, x_{2t-h})\Big]+0\notag\\
&=\frac{\delta_2}{\sqrt{n}} + \phi^h\E\Big[\varepsilon_{t-h}d(\widehat{Y}_{t-h},x_{1t-h}, x_{2t-h})\Big],\notag
\end{align}
where we have used \eqref{eq:(p.3)} in the third step, and $\mu\delta_1=\delta_2$ in the fourth. Exploiting the mutual independence of the disturbances and the fact that $\E[\varepsilon_{t-h}^3]=0$, we get that
\begin{align}
	\E&\Big[\varepsilon_{t-h}d(\widehat{Y}_{t-h},x_{1t-h}, x_{2t-h})\Big]\notag\\
	&\overset{\eqref{eq:(p.3)}}{=}\E\Big[\varepsilon_{t-h}(\varepsilon_{t-h}+\widehat{\varepsilon}_{t-h}-\varepsilon_{1,t-h-1})^2\Big] - \E\Big[\varepsilon_{t-h}(\mu(1-\phi) + \varepsilon_{t-h} + \widehat{\varepsilon}_{t-h})^2\Big]\notag\\
	&\hspace{0.22cm}=0-2\E\Big[\varepsilon_{t-h}^2\big(\mu(1-\phi) + \widehat{\varepsilon}_{t-h}\big)\Big]\notag\\
	&\hspace{0.22cm}=-2\mu(1-\phi)\sigma_{\varepsilon}^2.\label{eq:p.3.4}
\end{align}
Thus, 
\[
	\E\Big[\widehat{Y}_{t-1}d(\widehat{Y}_{t-h},x_{1t-h}, x_{2t-h})\Big]=\frac{\delta_2}{\sqrt{n}}-2\phi^{h}\mu(1-\phi)\sigma_{\varepsilon}^2
\]
for $h>1$. For $h=1$, $\widehat{\varepsilon}_{t-1}$ and $d(\widehat{Y}_{t-h},x_{1t-h}, x_{2t-h})$ in \eqref{eq:h=1} are not independent, giving rise to the additional term
\begin{align}
	\E&\Big[\widehat{\varepsilon}_{t-1}d(\widehat{Y}_{t-1},x_{1t-1}, x_{2t-1})\Big]\notag\\
	&\overset{\eqref{eq:(p.3)}}{=}\E\Big[\widehat{\varepsilon}_{t-1}\Big\{(\varepsilon_{t-1}+\widehat{\varepsilon}_{t-1}-\varepsilon_{1,t-2})^2 - (\mu(1-\phi) + \varepsilon_{t-1} + \widehat{\varepsilon}_{t-1})^2\Big\}\Big]\notag\\
	&\hspace{0.23cm}=\E\Big[\widehat{\varepsilon}_{t-1}\Big\{\widehat{\varepsilon}_{t-1}^2+2\widehat{\varepsilon}_{t-1}(\varepsilon_{t-1}-\varepsilon_{1,t-2}) + (\varepsilon_{t-1}-\varepsilon_{1,t-2})^2\Big\}\Big] \notag\\
	&\hspace{2cm}-\E\Big[\widehat{\varepsilon}_{t-1}\Big\{\widehat{\varepsilon}_{t-1}^2+2\widehat{\varepsilon}_{t-1}\big(\varepsilon_{t-1}+\mu(1-\phi)\big) + \big(\varepsilon_{t-1}+\mu(1-\phi)\big)^2\Big\}\Big]\notag\\
	&\hspace{0.23cm}=-2\E\Big[\widehat{\varepsilon}_{t-1}^2\big(\varepsilon_{t-1}+\mu(1-\phi)\big)\Big]\notag\\
	&\hspace{0.23cm}=-2\mu(1-\phi)\sigma_{\widehat{\varepsilon}}^2.\label{eq:p.4+}
\end{align}
So for general $h\geq1$, we obtain that
\[
	\E\Big[\widehat{Y}_{t-1}d(\widehat{Y}_{t-h},x_{1t-h}, x_{2t-h})\Big]=\frac{\delta_2}{\sqrt{n}}-2\phi^{h}\mu(1-\phi)\big(\sigma_{\varepsilon}^2+\sigma_{\widehat{\varepsilon}}^2I_{\{h=1\}}\big).
\]
Combining this with \eqref{eq:Yd}, it follows that
\begin{multline*}
	\E\Big[\Big(\widehat{Y}_{t-1}d(\widehat{Y}_{t},x_{1t}, x_{2t})-\frac{\delta_2}{\sqrt{n}}\Big)\Big(d(\widehat{Y}_{t-h},x_{1t-h}, x_{2t-h})-\frac{\delta_1}{\sqrt{n}}\Big)\Big]\\
	=-2\mu(1-\phi)\big(\sigma_{\varepsilon}^2+\sigma_{\widehat{\varepsilon}}^2I_{\{h=1\}}\big)\phi^{h}\frac{\delta_1}{\sqrt{n}}=:c_1\phi^{h}\frac{\delta_1}{\sqrt{n}}.
\end{multline*}

Finally, for the lower-right element, we get that 
\begin{align}
\E&\Big[\Big(\widehat{Y}_{t-1}d(\widehat{Y}_{t},x_{1t}, x_{2t})-\frac{\delta_2}{\sqrt{n}}\Big)\Big(\widehat{Y}_{t-h-1}d(\widehat{Y}_{t-h},x_{1t-h}, x_{2t-h})-\frac{\delta_2}{\sqrt{n}}\Big)\Big]\notag\\
&= \E\Big[\widehat{Y}_{t-1}d(\widehat{Y}_{t},x_{1t}, x_{2t})\widehat{Y}_{t-h-1}d(\widehat{Y}_{t-h},x_{1t-h}, x_{2t-h})\Big] - \frac{\delta_2}{\sqrt{n}}\E\Big[\widehat{Y}_{t-1}d(\widehat{Y}_{t},x_{1t}, x_{2t})\Big]\notag\\
&\hspace{7.5cm} - \frac{\delta_2}{\sqrt{n}}\E\Big[\widehat{Y}_{t-h-1}d(\widehat{Y}_{t-h},x_{1t-h}, x_{2t-h})\Big] + \frac{\delta_2^2}{n}\notag\\
&=\E\Big[d(\widehat{Y}_{t},x_{1t}, x_{2t})\Big]\E\Big[\widehat{Y}_{t-1}\widehat{Y}_{t-h-1}d(\widehat{Y}_{t-h},x_{1t-h}, x_{2t-h})\Big] - \frac{\delta_2^2}{n}\notag\\
&= \frac{\delta_1}{\sqrt{n}}\E\Big[\widehat{Y}_{t-1}\widehat{Y}_{t-h-1}d(\widehat{Y}_{t-h},x_{1t-h}, x_{2t-h})\Big]-\frac{\delta_2^2}{n}.\label{eq:YYd}
\end{align}
Using \eqref{eq:Yhat decomp}, we obtain for $h>1$ that
\begin{align}
\E&\Big[\widehat{Y}_{t-1}\widehat{Y}_{t-h-1}d(\widehat{Y}_{t-h},x_{1t-h}, x_{2t-h})\Big]\notag\\
&=\E\Big[\Big\{\mu(1-\phi^h) + \phi^{h}Y_{t-h-1}+ \phi^{h}\varepsilon_{t-h} + \sum_{i=1}^{h-1}\phi^{i}\varepsilon_{t-i} +\widehat{\varepsilon}_{t-1}\Big\}\widehat{Y}_{t-h-1}d(\widehat{Y}_{t-h},x_{1t-h}, x_{2t-h})\Big]\label{eq:h=1.2}\\
&=\mu(1-\phi^h)\frac{\delta_2}{\sqrt{n}}+\phi^{h}\E\Big[Y_{t-h-1}\widehat{Y}_{t-h-1}d(\widehat{Y}_{t-h},x_{1t-h}, x_{2t-h})\Big] \notag\\
&\hspace{2.8cm} + \phi^{h}\E\Big[\varepsilon_{t-h}\widehat{Y}_{t-h-1}d(\widehat{Y}_{t-h},x_{1t-h}, x_{2t-h})\Big]+0\notag\\
&=\mu(1-\phi^h)\frac{\delta_2}{\sqrt{n}}+\phi^{h}\E\Big[Y_{t-h-1}(Y_{t-h-1}+\widehat{\varepsilon}_{t-h-1})\Big]\E\Big[d(\widehat{Y}_{t-h},x_{1t-h}, x_{2t-h})\Big] \notag\\
&\hspace{2.8cm} + \phi^{h}\E\Big[\widehat{Y}_{t-h-1}\Big]\E\Big[\varepsilon_{t-h}d(\widehat{Y}_{t-h},x_{1t-h}, x_{2t-h})\Big]\notag\\
&=\mu(1-\phi^h)\frac{\delta_2}{\sqrt{n}}+\phi^{h}\E\Big[Y_{t-h-1}^2\Big]\frac{\delta_1}{\sqrt{n}} +\phi^{h}\mu(-2)\mu(1-\phi)\sigma_{\varepsilon}^2\notag\\
&=\mu(1-\phi^h)\frac{\delta_2}{\sqrt{n}}+\phi^{h}\Big[\frac{\sigma_{\varepsilon}^2}{1-\phi^2}+\mu^2\Big]\frac{\delta_1}{\sqrt{n}} -2\phi^{h}\mu^2(1-\phi)\sigma_{\varepsilon}^2\notag\\
&=\mu\frac{\delta_2}{\sqrt{n}}+\phi^{h}\frac{\sigma_{\varepsilon}^2}{1-\phi^2}\frac{\delta_1}{\sqrt{n}}-2\phi^{h}\mu^2(1-\phi)\sigma_{\varepsilon}^2,\notag
\end{align}
where we have used \eqref{eq:p.3.4} in the fourth step, \eqref{eq:(p.30)} in the fifth, and $\mu\delta_1=\delta_2$ in the sixth. For $h=1$, $\widehat{\varepsilon}_{t-1}$ and $d(\widehat{Y}_{t-h},x_{1t-h}, x_{2t-h})$ in \eqref{eq:h=1.2} are not independent, giving rise to the additional term
\begin{align*}
\E\big[\widehat{\varepsilon}_{t-1}\widehat{Y}_{t-2}d(\widehat{Y}_{t-1}, x_{1t-1}, x_{2t-1})\big] &= \E[\widehat{Y}_{t-2}]\E\big[\widehat{\varepsilon}_{t-1}d(\widehat{Y}_{t-1}, x_{1t-1}, x_{2t-1})\big]\\
&=-2\mu^2(1-\phi)\sigma_{\widehat{\varepsilon}}^2,
\end{align*}
where we have used \eqref{eq:p.4+}. Hence, for general $h\geq1$,
\begin{multline*}
	\E\Big[\widehat{Y}_{t-1}\widehat{Y}_{t-h-1}d(\widehat{Y}_{t-h},x_{1t-h}, x_{2t-h})\Big] \\
	= \mu\frac{\delta_2}{\sqrt{n}}+\phi^{h}\frac{\sigma_{\varepsilon}^2}{1-\phi^2}\frac{\delta_1}{\sqrt{n}}-2\phi^{h}\mu^2(1-\phi)\big(\sigma_{\varepsilon}^2 + \sigma_{\widehat{\varepsilon}}^2I_{\{h=1\}}\big).
\end{multline*}
Combining with \eqref{eq:YYd} and once again using $\mu\delta_1=\delta_2$ gives
\begin{align*}
\E&\Big[\Big(\widehat{Y}_{t-1}d(\widehat{Y}_{t},x_{1t}, x_{2t})-\frac{\delta_2}{\sqrt{n}}\Big)\Big(\widehat{Y}_{t-h-1}d(\widehat{Y}_{t-h},x_{1t-h}, x_{2t-h})-\frac{\delta_2}{\sqrt{n}}\Big)\Big]\\
&=\frac{\delta_1}{\sqrt{n}}\Bigg\{\mu\frac{\delta_2}{\sqrt{n}}+\phi^{h}\frac{\sigma_{\varepsilon}^2}{1-\phi^2}\frac{\delta_1}{\sqrt{n}}-2\phi^{h}\mu^2(1-\phi)\big(\sigma_{\varepsilon}^2+\sigma_{\widehat{\varepsilon}}^2I_{\{h=1\}}\big)\Bigg\}-\frac{\delta_2^2}{n}\\
&=\Bigg\{\frac{\sigma_{\varepsilon}^2}{1-\phi^2}\frac{\delta_1}{\sqrt{n}}-2\mu^2(1-\phi)\big(\sigma_{\varepsilon}^2+\sigma_{\widehat{\varepsilon}}^2I_{\{h=1\}}\big)\bigg\}\phi^{h}\frac{\delta_1}{\sqrt{n}}\\
&=:c_{2}\phi^{h}\frac{\delta_1}{\sqrt{n}}.
\end{align*}
Overall, we obtain that
\[
	\E\Big[\Big(\widehat{Z}_{n,t}-\frac{\delta}{\sqrt{n}}\Big)\Big(\widehat{Z}_{n,t-h}-\frac{\delta}{\sqrt{n}}\Big)^\prime\Big]=\begin{pmatrix}0 & 0 \\ c_{1}\phi^{h}\frac{\delta_1}{\sqrt{n}} & c_{2}\phi^{h}\frac{\delta_1}{\sqrt{n}}\end{pmatrix}.
\]
Using properties of the geometric series, it follows that
\[
	\frac{1}{n}\sum_{h=1}^{n}\sum_{t=h+1}^{n} c_{i}\phi^{h}\frac{\delta_1}{\sqrt{n}}=c_i\frac{\delta_1}{\sqrt{n}}\Bigg\{\sum_{h=1}^{n}\frac{n-h}{n}\phi^{h}\Bigg\}=o(1),
\]
implying
\[
	(II)=\frac{1}{n}\sum_{h=1}^{n}\sum_{t=h+1}^{n}\E\Big[\Big(\widehat{Z}_{n,t}-\frac{\delta}{\sqrt{n}}\Big)\Big(\widehat{Z}_{n,t-h}-\frac{\delta}{\sqrt{n}}\Big)^\prime\Big]=o(1).
\]
All in all, \eqref{eq:Omega} follows.
\end{proof}

\singlespacing
\putbib[thebib]
\end{bibunit}

\end{document}